\documentclass[reqno]{amsart}

\usepackage{a4wide}
\usepackage{color}
\usepackage{verbatim}
\usepackage[usenames,dvipsnames]{xcolor}

\newcommand{\nml}{\mathbf{n}}
\newcommand{\ee}{\mathbf{e}}

\newcommand{\del}{\delta}
\newcommand{\sphere}{\mathbb{S}^{d-1}}
\newcommand{\setR}{\mathbb{R}}
\newcommand{\setC}{\mathbb{C}}

\newcommand{\E}{\mathbb{E}}
\newcommand{\J}{\mathbb{I}}
\newcommand{\CB}{\mathcal{B}}

\newcommand{\sod}{{\mathrm{SO}(d)}}

\newcommand{\czero}{C^0(\sod)}

\newcommand{\np}{{n'}}

\def\CB {\mathcal{B}}

\def\CF { \mathcal{F}}

\def\CL { \mathcal{L}}

\def\CO{\mathcal{O}}
\def\CP { \mathcal{P}}
\def\CS {\mathcal{S}}

\def\a {\alpha}

\def\s {\sigma}
\def\eps {\epsilon}

\def\a {\alpha}

\def\J {\mathbb{I}}

\def\N {\mathbb{N}}

\def\E {\mathbb{E}}
\def\RE {\mathbb{R}}

\def\D {\mathbb{D}}
\def\P {\mathbb{P}}

\newcommand{\soda}{{\mathrm{SO}(d-1)}}
\newcommand{\sodk}{{\mathrm{SO}(k)}}

\newcommand{\haar}{\mathfrak{H}}
\newcommand{\rot}{\mathcal{O}}
\newcommand{\eed}{\mathbf{e}_d}
\newcommand{\eins}{\mathbf{1}}
\newcommand{\rota}{\mathcal{O}}
\newcommand{\law}{\mathrm{Law}}

\newcommand{\mix}{\mathfrak{m}}
\newcommand{\meas}{\mathcal{P}(\setR^d)}
\newcommand{\deq}{\stackrel{\CL}{=}}
\newcommand{\pby}{\mathbb{P}}
\newcommand{\uuu}{\mathfrak{u}_{\mathbb S}}

\newcommand{\dn}{\mathrm{d}}
\newcommand{\dd}{\,\mathrm{d}}

\newcommand{\spr}[2]{\langle #1,#2 \rangle_{L^2}}

\def \tpb { \mathfrak{\tilde B}} 
\def \pb  {\mathfrak{B}}  
\def \be {  \varpi}
\def  \ww { \beta}

\def \gg {\mathfrak{b}}

\newtheorem{definition}{Definition}
\newtheorem{theorem}[definition]{Theorem}
\newtheorem{lemma}[definition]{Lemma}
\newtheorem{proposition}[definition]{Proposition}

\newtheorem{remark}{Remark}

\usepackage{amsmath}
\begin{document}

\title{Infinite energy solutions to inelastic homogeneous Boltzmann equation}

\date{\today}

\author{Federico Bassetti}
\address{F.B., Dipartimento di Matematica, Universit\`a di Pavia, Pavia, Italy}
\author{Lucia Ladelli}
\address{L.L., Dipartimento di Matematica, Politecnico di Milano, Milano, Italy}
\author{Daniel Matthes}
\address{D.M., Zentrum Mathematik, Technische Universit\"at M\"unchen, Germany}


\begin{abstract} 
  This paper is concerned with the existence, shape and dynamical stability of infinite-energy equilibria 
  for a general class of spatially homogeneous kinetic equations in space dimensions $d\ge3$.
  Our results cover in particular Bobyl\"ev's model for inelastic Maxwell molecules.
  First, we show under certain conditions on the collision kernel, that there exists an index $\alpha\in(0,2)$
  such that the equation possesses a nontrivial stationary solution, 
  which is a scale mixture of radially symmetric $\alpha$-stable laws.
  We also characterize the mixing distribution as the fixed point of a smoothing transformation.
  Second, we prove that any transient solution that emerges from the NDA of some (not necessarily radial symmetric) $\alpha$-stable distribution
  converges to an equilibrium.
  The key element of the convergence proof is an application of the central limit theorem 
  to a representation of the transient solution as a weighted sum of i.i.d. random vectors.
\end{abstract}

\maketitle

\section{Introduction}
\subsection{The equation}
In this paper, we analyzed the long-time asymptotic of the velocity distribution in kinetic models 
for spatially homogeneous inelastic Maxwellian molecules \cite{Bobylev99}, and certain generalizations. We assume that the space dimension
$d$ is at least two, with the physical situation $d=3$ being the most 
interesting choice. Under the \emph{cut off assumption} and after proper normalization of the collision frequency,
the evolution equation for the time-dependent velocity distribution $\mu:\setR_+\to\meas$ 
is given by 
\begin{equation}
  \label{inelasticboltzmann1}
  \left \{
    \begin{array}{l}
      \partial_t \mu(t) + \mu(t) = Q_+( \mu(t),\mu(t) ) \qquad (t > 0) \\
      \mu(0)=\mu_0\\
    \end{array}
  \right .
\end{equation}
where the \emph{collisional gain operator} $Q_+$ has the weak formulation
\begin{align*}
  \int_{\setR^d} \varphi(v)Q_+( \mu,\mu )\dd v
  = \int_{\setR^d\times \setR^d} \E\Big[\frac{\varphi(v')+\varphi(v'_*)}{2}\Big]\dd\mu( v) \dd\mu( v_*),
  \quad \text{for all $\varphi\in C^0_b(\setR^d)$.}
\end{align*}
Above, the expectation $\E$ is taken for the \emph{post-collisional velocities} $v',v_*'$,
which are random vectors whose distribution is determined from the \emph{pre-collisional velocities} $v,v_*$ by means of collision rules.
Later, we formulate our hypotheses on the collision rules under which we are able to prove existence and stability of stationary solutions.
For the original inelastic Maxwell molecules from \cite{Bobylev99}, these rules read as
\begin{align}
  \label{eq:maxwellrules}
  \begin{split}
    v' &= \frac12(v+w) +\Big (\frac{\delta}{2}(v-w) +\frac{1-\delta}{2}  |v-v_*|  \nml \Big), \\
    v_*' &=  \frac{1}{2}(v+w) -\Big(\frac{\delta}{2}(v-w) +\frac{1-\delta}{2}  |v-v_*| \nml \Big),    
  \end{split}
\end{align}
where $\delta\in(0,1/2)$ is the modulus of inelasticity, 
and $\nml$ is a random unit vector of prescribed distribution on the unit sphere $\sphere$:
there is a properly normalized (see \eqref{eq:ggnormal} below) density function (\emph{cross section}) $\gg\in L^1(-1,1)$,
such that $\nml$ has law
\begin{align}
  \label{eq:cross}
  \gg\Big(\sigma\cdot\frac{(v-v_*)}{|v-v_*|}\Big)\uuu(\dn\sigma),
\end{align}  
where $\uuu$ is the normalized surface measure (uniform probability) on $\sphere$.
The characteristic property of Maxwellian molecules -- in contrast to more general ideal gases --
is that this density does \emph{not} explicitly depend on the norm $|v-v_*|$ of the relative velocity.
This property allows to restate \eqref{inelasticboltzmann1} as an evolution equation for the characteristic function 
$\hat\mu(t;\xi)=\int \exp(i\xi\cdot v)\mu(t;\dn v)$ of $\mu(t)$,
\begin{align}
  \label{eq.boltz1}
  \partial_t\hat\mu(t) + \hat\mu(t) = \widehat{Q_+}[\hat\mu(t),\hat\mu(t)],
\end{align}
with an explicit form of the Fourier transformed collision operator:
there are non-negative random variables $r^\pm$ and random rotations $R^\pm$ in $\sod$, 
such that
\begin{align}
  \label{eq:miracle}
  \widehat{Q_+}[\hat\mu,\hat\mu](\rho\rota\eed) = \E\big[\hat\mu(\rho r^+\rota R^+\eed)\hat\mu(\rho r^-\rota R^-\eed)\big]
\end{align}
holds for all $\rho\in\setR_+$ and all $\rota\in\sod$.
This special form of $Q_+$ is of crucial importance for our analysis of \eqref{inelasticboltzmann1} by probabilistic tools.
The existence of such a representation \eqref{eq:miracle} is by no means obvious.
A similar expression has been given for the collision operator modelling fully elastic Maxwell molecules recently \cite{DoRe2}.
For inelastic molecules, it is proven in Proposition \ref{lemmaProbrepK} below.

\subsection*{Notice:}
In the following, we assume that the reader is familar with basic notions of the central limit theorem,
in particular with the L\'{e}vy representation of multi-dimensional $\alpha$-stable distributions and their {\it normal domain of attraction} (NDA).
A brief introduction to this topic is included in Appendix \ref{app:stable}.

\subsection{Related results}
In the rich literature on long-time asymptotics for \eqref{inelasticboltzmann1}, 
both solutions with finite (kinetic) energy, that is, 
\begin{align*}
  \int_{\setR^d} |v|^2\mu(t;\dn v) < \infty \quad \text{for all $t>0$},
\end{align*}
and with infinite energy have been studied.
In order to relate our own results to the existing literature, 
we briefly recall a small selection of results on convergence to equilibrium for elastic and inelastic Maxwell molecules;
the following summary is focussed on weak convergence results under minimal hypotheses on the initial conditions.
\begin{itemize}
\item \emph{Finite energy solutions for fully elastic collisions.}
  The only stationary solutions of finite energy to the fully elastic Maxwell model \cite{Bobylev88} are Gaussians, 
  and these attract \emph{all} solutions of finite energy.
  This is known as Tanaka's theorem \cite{Tanaka}.
  Various simple proofs are available, see e.g.\ \cite{ToscVill}.
\item \emph{Infinite energy solutions for fully elastic collisions.}
  The elastic Maxwell model does not admit stationary solutions of infinite energy \cite{CaGaRe}. 
  However, Bobyl\"ev and Cercignani \cite{BoCe} have identified for every $\alpha\in(0,2)$ a family
  of \emph{self-similar} solutions 
  for which the $\alpha$th moment is marginally divergent.
  These self-similar solutions converge vaguely to zero as time goes to infinity, i.e., the velocities concentrate at infinity.
  It has been shown recently \cite{Cannone} that the self-similar solutions for a given $\alpha$ attract all transient solutions (of infinite energy)
  whose initial condition's characteristic function $\hat\mu_0$ satisfies
  \begin{align}
    \label{eq:karchhypo}
    \lim_{|\xi|\to0}\frac{\hat\mu_0(\xi)-1}{|\xi|^\alpha} = K
    \quad\text{for some $K<0$}.
  \end{align}
\item \emph{Finite energy solutions for inelastic collisions.}
  Inelastic Maxwellian molecules lose kinetic energy in every collision.
  If the energy is finite initially, then it converges to zero exponentially fast in time \cite{Bobylev99}.
   As was conjectured by Ernst and Brito \cite{ErnsBrit}, this collapse happens in a self-similar way. More precisely, there
is a time-dependent rescaling of the velocity variable such that the
rescaled Boltzmann equation possesses a family of non-trivial
stationary solutions, the so-called homogeneous cooling states.
It has further been proven \cite{BiCaTo2,BoCeter,BoCeTo,BolleyCarrillo} that any solution of finite energy
to the rescaled equation eventually converges towards one of these
cooling states.
\item \emph{Infinite energy solutions for inelastic collisions.}
  This case has received less attention than the aforementioned situations.
  Some results are available for the \emph{inelastic Kac model} \cite{PulvTosc}, which is a one-dimensional caricature of inelastic Maxwell molecules:
  for each inelastic Kac model, there is precisely one $\alpha\in(0,2)$, 
  such that the symmetric $\alpha$-stable laws are stationary solutions and attract all transient solutions that start in their respective NDA \cite{BaLaRe}.
  A generalization of this result has been obtained by the authors \cite{BaLaMa} for Kac-type models with more complicated collisions and a richer class of stationary states.
  A related generalization \cite{BoCeGa2,BoCeGa} also covers the case of radially symmetric solutions to the inelastic Kac model in multiple space dimensions.
  The existence of a family of stationary solutions is proven, and the $\alpha$th moment of that solutions is marginally divergent,
  where $\alpha\in(0,2)$ is specific for the considered model.
  It is shown that the self-similar solutions attracts all radially symmetric solutions whose initial condition satisfies a condition 
  that is slightly more restrictive than \eqref{eq:karchhypo} as above. Using the results contained in  \cite{BaLa}, it can be proved that the same conclusions hold under condition \eqref{eq:karchhypo}.
\end{itemize}
Various of these fundamental weak convergence results have been
made quantitative (e.g.\ in terms of estimates on convergence \emph{rates}) 
and improved qualitatively (by proving e.g.\ convergence in \emph{strong} topologies).
Naturally, such improvements require additional hypotheses on the initial data (like higher moments or finite entropy)
and are not of interest here.
We refer the reader to the reviews \cite{CarrTosc,Villani}, 
and to the more recent results on self-similar asymptotics for inelastic Maxwell molecules \cite{CarlCarrCarv} and for inelastic hard spheres \cite{Mischler3}.

\subsection{Results and Method} 
In the present paper, 
we give a refined analysis of infinite energy solutions for kinetic equations with collision kernel 
of the form \eqref{eq:miracle} in general, and for inelastic Maxwell molecules in particular.
We show the existence of a family of stationary solutions
and we give a representation for them as scale mixtures of radially symmetric $\alpha$-stable laws.
Our main result is the dynamic stability of stationary solutions under assumptions on the initial conditions that we expect to be minimal.
The full statement is given in Theorem \ref{thm.main}.
In the special case of inelastic Maxwellian molecules, it reduces to the following.
\begin{theorem}
  \label{thm.simple}
  Consider equation \eqref{inelasticboltzmann1} with collision rules \eqref{eq:maxwellrules},
  where $\delta\in(0,1/2)$, and the unit vector $\nml$ has law \eqref{eq:cross} with cross section $\gg$,
  which is such that
    \begin{align}
      \label{eq:ggnormal}
      \int_{{-1}}^1 \gg(z)\sqrt{(1-z^2)^{d-3}}\dd z = \int_0^1 \sqrt{z^{{-1}}(1-{z})^{d-3}}\dd z.
    \end{align}
  
  Then there is a unique exponent $\alpha\in(0,2)$ and a probability measure $\mix$ on $\setR_+$
  -- both computable from $\delta$ and $\gg$ in principle --
  such that the following is true.

  A one-parameter family $(\mu_\infty^c)_{c>0}$ of stationary solutions to \eqref{inelasticboltzmann1}
  is given in terms of their characteristic functions $\hat\mu_\infty^c$ by
  \begin{align*}
    \hat\mu^c_\infty(\xi) = \int_{\setR_+} \exp\big(-cu|\xi|^\alpha\big)\mix(\dn u) \quad \text{for all $\xi\in\setR^d$}.
  \end{align*}
If $\mu_0$ belongs to the NDA of a full $\alpha$-stable distribution 
  {\rm (}centered, if $\alpha>1$, and an additional condition is needed if $\alpha=1$ - see  \eqref{hy-alpha=1BIS} in Section \ref{sec.general} {\rm)},
  then the corresponding solution $\mu$ to \eqref{inelasticboltzmann1} converges weakly to a stationary solution $\mu_\infty^c$,
  where $c\in\setR_+$ is computable in terms of $\mu_0$.
  In particular, the $\mu_\infty^c$ are the only stationary solutions that belong to the NDA of some $\alpha$-stable distribution on $\setR^d$.
\end{theorem}
Apparently, these are the first results on the stability of stationary solutions in the inelastic Maxwell model \emph{without} the assumption of radial symmetry.
Indeed, it seems that the approach to derive long-time asymptotics directly
from contraction estimates on the Fourier transform of the transient solutions 
like in \cite{BoCeGa} or \cite{Cannone}, needs an hypothesis on the initial datum of the form \eqref{eq:karchhypo}. 
This hypothesis is significantly stronger than ours, as can be seen from the
characterization of NDAs by means of characteristic functions, see e.g. \cite{AaronsonDenker}.
For instances, \eqref{eq:karchhypo} implies that $\mu_0$ belongs to the NDA of a \emph{radially
symmetric} $\alpha$-stable law, which further implies that $\mu_0$ is radially
symmetric "asymptotically" on the complement of large balls. Hence, our condition
that $\mu_0$ belongs to the NDA of \emph{some} full $\alpha$-stable law is
much weaker.
In fact, we expect that the NDA is a \emph{sharp} characterization of the basin of attraction for the kinetic equation
in the sense that all other transient solutions either concentrate at the origin or vaguely converge to zero as time tends to infinity.

The key element in our proof is a probabilistic representation of the solution $\mu$ to \eqref{inelasticboltzmann1}.
First, $\mu$ can be written as a Wild sum,
\begin{align}
  \label{Wild}
  \mu(t) = \sum_{n=0}^\infty e^{-t}(1-e^{-t})^n\mu_n.
\end{align}
Now each of the $\mu_n\in\meas$ is the law of a random vector $V_n$ in $\setR^d$, and the $V_n$ are characterized as follows.
There is an array $(\beta_{k,n},O_{k,n})_{1\le k\le n+1}$ of non-negative random numbers $\beta_{k,n}$ and random rotations $O_{k,n}$ in $\sod$,
such that for every $\rota$ in $\sod$
\begin{align}
  \label{eq:probrep}
  (\rota\eed)\cdot V_n \stackrel{\CL}{ =} \sum_{k=1}^{n+1}(\beta_{k,n}\rota O_{k,n}\eed)\cdot X_k,
\end{align}
where the $X_k$ are i.i.d. (independent and identically distributed) random variables, independent of $\beta_{k,n}$ and $O_{k,n}$, with distribution $\mu_0$.
Now, the techniques pertaining to the central limit theorem are adapted to conclude weak convergence of the $\mu_n$ to some $\mu_\infty^c$,
and this implies via \eqref{Wild} weak convergence of $\mu(t)$ to the same limit as $t\to\infty$.

The general idea of a probabilistic representation of Boltzmann like equations goes back essentially to McKean \cite{McKean1966,McKean1967},
who applied it to the Kac equation, a caricature of the homogeneous Boltzmann equation in dimension $d=1$. 
The idea has since then been extended and refined,
for instance in \cite{DolGabReg,dolera2,GabettaRegazziniCLT,GabettaRegazziniWM} (for the Kac equation) 
and \cite{BaLa,BaLaMa,BaLaRe} (for various one-dimensional  Kac-type kinetic equations).

The extension to dimension $d>1$ is by no means straightforward. 
Only in the very recent paper \cite{DoRe2}, Dolera and Regazzini derived 
a suitable probabilistic representation of the solution of the homogeneous Boltzmann equation in dimension $d=3$,
using particular coordinates on $\setR^3$ and its rotation group.
Here, we extend the Dolera-Regazzini probabilistic representation 
to equation \eqref{eq.boltz1} with kernels of the form \eqref{eq:miracle},
in arbitrary dimensions $d\ge3$.
Our probabilistic representation is summarized in Proposition \ref{probrep},
which should be an interesting result in itself.

\subsection{Plan of the paper}
In Section 2 below we state our hypotheses and formulate the main result 
about the general kinetic model with collision kernel of type \eqref{eq:miracle}.
We also introduce the main tool for the proof: the probabilistic representation of transient solutions.
In Section 3, we prove that inelastic Maxwell molecules, 
that is \eqref{inelasticboltzmann1} with collision rules \eqref{eq:maxwellrules} 
and an arbitrary choice for the density $\gg$ of the cross section \eqref{eq:cross},
fit into the general framework provided in the previous section.
Sections 4 and 5 contain the proof of the main result, which is naturally divided into two parts:
Section 4 is concerned with contraction estimates on a random walk in the rotation group,
which is induced by our probabilistic representation.
In Section 5, we apply the machinery of the central limit theorem to the representation \eqref{eq:probrep}
to obtain the long-time asymptotics of transient solutions to \eqref{inelasticboltzmann1}.
The Appendix contains a summary of various results on $\alpha$-stable distributions
that are relevant to our proofs.

\section{An abstract Boltzmann-like equation}
In this section, we formulate our hypotheses and state our results 
for the general kinetic equation \eqref{eq.boltz1} with collision kernel \eqref{eq:miracle}.
We will see in Section \ref{sec.maxwell} that inelastic Maxwell molecules fall into this model class,
so Theorem \ref{thm.simple} from the introduction follows as a corollary 
from the general Theorem \ref{thm.main} below.

\subsection{Notations}
Denote by $\sod$ the usual orientation-preserving rotation group in $\setR^d$,
and -- by abuse of notation -- by $\soda$ its subgroup that acts on $\setR^{d-1}\subset\setR^d$ only,
i.e., that leaves the ``last'' unit vector $\eed:=(0,\ldots,0,1)\in\setR^d$ invariant. 
As usual, $\deq$ means \emph{equality in law}.
For two probability measures $\pb$ and $\pb'$ on $\sod$,
define their \emph{convolution} $\pb\star\pb'$ as the probability measure with
\begin{align*}
  \pb\star\pb'(A):=\int_\sod \pb(R^TA)\dd\pb'(R)\quad\text{for every measurable set $A$ in $\sod$}.
\end{align*}
Accordingly, we define powers $\pb^{\star2}=\pb\star\pb$ etc. 
Finally, let $\haar$ be the Haar measure on $\sod$.

\subsection{Main assumptions and results}\label{sec.general}
Let $(r^-,r^+,R^-,R^+)$ be a random element defined on a suitable probability space $(\Omega,\CF,\pby)$ 
taking values in $\setR_+ \times \setR_+ \times \sod \times \sod$ 
and denote by $\E$ the expectation with respect to $\pby$.
Our assumptions on the law of $(r^-,r^+,R^-,R^+)$ are the following
\begin{itemize}
\item[(H1)] \textit{For any rotations $\rota_1,\rota_2\in\sod$ such that $\rota_1 \eed=\rota_2 \eed$},
  we have 
  \begin{align*}
    ( \rota_1 r^- R^- \eed, \rota_1 r^+ R^+ \eed )\deq ( \rota_2 r^- R^- \eed, \rota_2 r^+ R^+ \eed ).
  \end{align*}
\item[(H2)] \textit{There are an $\alpha\in(0,2)$ and a $\gamma>1$ such that}
  \begin{align*}
    \E[(r^-)^\alpha+(r^+)^\alpha]=1,\quad\text{and}\quad\E[(r^-)^{\alpha\gamma}+(r^+)^{\alpha\gamma}]<1,    
  \end{align*}
  \textit{and, in addition, $\pby\{r^->0\}+\pby\{r^+>0\}>1$}.
\end{itemize}
For later reference, we introduce the (convex) function $\CS:[0,\infty)\to[-1;\infty]$ by
\begin{align}
  \label{eq:CS}
  \CS(s) := \E\big[(r^+)^s+(r^-)^s\big]-1.
\end{align}
Then (H2) can be rephrased in the form that
\begin{align*}
  \CS(\alpha) = 0 \quad \text{and} \quad \CS(\alpha\gamma)<0 \quad \text{for some $\alpha\in(0,2)$ and $\gamma>1$}.
\end{align*}
Under hypothesis (H2), the following defines probability measures $\pb^+$ and $\pb^-$ on $\sod$:
\begin{align}
  \label{eq:BB}
  \int_\sod f(R) \pb^\pm(\dn R)= \frac{\E\big[ (r^\pm)^\alpha f\big(R^\pm \big)\big]}{\E[(r^\pm)^\alpha]}
  \quad\text{for all $f\in C^0_b(\sod)$}.
\end{align}
In addition to (H1)-(H2), we shall assume further:
\begin{itemize}
\item[(H3)]
  \textit{The probability measures $\pb^\pm$ are non-singular with respect to the Haar measure, i.e.
    they have  a non-trivial absolutely continuous component with respect to $\haar$. }
\end{itemize}
Before stating the general form of our main result,
we briefly comment on the role of assumptions (H2) and (H3). 
Assumption (H2) is a classical hypothesis which guarantees 
the existence of a (unique up to scaling) fixed point of the \emph{smoothing transformation} associated with $(r^-,r^+)$.
The respective result is the following.
\begin{proposition}[see \cite{alsmeyer,DurrettLiggett1983}]\label{PropDurretLigget}
  Under assumption {\rm (H2)} there is a unique probability measure $\mix$ on $[0,+\infty)$ with $\int u\mix(\dn u)=1$
  whose characteristic function $\hat\mix$ satisfies 
  \begin{equation}\label{fixedpointDL}
    \hat\mix(\xi) = \E[ \hat\mix((r^-)^\alpha \xi) \hat\mix((r^+)^\alpha \xi)]
    \quad\text{for all $\xi \in \setR$}. 
  \end{equation}
  Moreover, for every $p>1$, 
  $\int u^p\mix(\dn u)<+\infty$ if and only if $\E[(r^-)^{p\alpha}+(r^+)^{p\alpha}]<1$.
\end{proposition}
Assumption (H3) entails the convergence of the $n$-fold convolution $(\pb^\pm)^{\star n}$ to the Haar measure $\haar$. 
See e.g.\ \cite{Bhattacharya} for a proof of exponentially fast convergence in total variation.
We only need a corollary of that result, which is formulated in Proposition \ref{Bhatta2}.
\medskip

With the notations and preliminary results at hand, we can formulate our main theorem.
\begin{theorem}\label{thm.main}
 For a given random element $(r^-,r^+,R^-,R^+)$, 
    define a collision operator $Q_+$ by means of \eqref{eq:miracle}.
    Assume that there is an $\alpha\in(0,2)$ such that hypotheses {\rm (H1)-(H3)} hold,
    and consider the initial value problem \eqref{inelasticboltzmann1} 
    with an initial condition $\mu_0$ that belongs to the NDA of a full $\alpha$-stable distribution with L\'evy measure $\phi$.
  If $\alpha>1$, assume also that $\mu_0$ is centered,
  while if $\alpha=1$, assume that there is some $\gamma_0 \in \setR^d$ with
  \begin{equation}\label{hy-alpha=1BIS}
    \lim_{R \to +\infty} \sup_{\sigma\in\sphere} \Big | \int_{ -R<\sigma\cdot v \leq R } \sigma\cdot v\,\mu_0(\dn v) - \sigma\cdot\gamma_0 \Big| =0 .
  \end{equation}
Then the unique solution $\mu(t)$ to \eqref{inelasticboltzmann1} 
  converges weakly to the probability distribution $\mu_\infty^c$ 
  that has the characteristic function
  \[
  \hat\mu_\infty^c(\xi) = \int_{[0,\infty)} \exp\big( -c u |\xi|^\alpha\big) \mix(\dn u)
  \quad \text{for all $\xi \in \setR^d$},
  \]
  where the probability measure $\mix$ is defined in {\rm Proposition \ref{PropDurretLigget}},
  and
  \[
  c=\frac{1}{\Gamma(\alpha)\sin(\pi\alpha/2)} \int_{\sphere}\int_{\{y: y\cdot\sigma> 1\}} \phi(\dn y) \uuu(\dn\sigma).
  \]
 In particular, the $\mu_\infty^c$ are the only stationary solutions of \eqref{inelasticboltzmann1}
 that belong to the NDA of some $\alpha$-stable distribution on $\setR^d$.
\end{theorem}
The proof of Theorem \ref{thm.main} is given in Section \ref{sec.proof}.

\subsection{A probabilistic representation}
As already mentioned in the introduction, 
the key element in our proof of Theorem \ref{thm.main} is a suitable stochastic representation of $\mu(t)$ 
connected to a randomly weighted sum of i.i.d.\ random vectors.
This probabilistic representation enables us to study the long-time asymptotics of $\mu(t)$ 
by methods related to the central limit theorem.

The starting point is the Wild sum representation \eqref{Wild} of solutions to \eqref{inelasticboltzmann1}.
Equivalently, the time-dependent characteristic function $\hat\mu$ satisfying \eqref{eq.boltz1}
can be written as
\begin{align}
  \label{Wild2zero}
  \hat\mu(t;\xi) = \sum_{n=0}^\infty e^{-t}(1-e^{-t})^n\hat\mu_n(\xi),
\end{align}
where the charcteristic functions $\hat\mu_n$ of the probability measures $\mu_n$ are defined 
inductively from the initial condition $\hat\mu_0$ as follows:
\begin{align}
  \label{Wild2}
  {\hat \mu}_{n+1}(\xi)=\frac{1}{n+1} \sum_{k= 0}^{n} \hat Q_+({\hat \mu}_k,{\hat \mu}_{n-k})(\xi)
  \quad \text{for all $n=0,1,2,\ldots$}
\end{align}
The probabilistic representation we introduce now gives a meaning to the measures $\mu_n$
--- or rather, to their characteristic functions $\hat\mu_n$ ---
in terms of randomly weighted sums of i.i.d.\ random vectors.

The setup is the following:
Let the following be given on a sufficiently large probability space $(\Omega, \CF,\pby)$:
\begin{itemize}
\item a sequence of independent random variables  $(\ell_n)_{n \geq 1}$ 
  such that each $\ell_n$ is uniformly distributed on $\{1,\ldots,n\}$;
\item  a sequence of i.i.d. random vectors  $(X_j)_{j\geq1}$  with distribution $\mu_0$;
\item  a sequence of i.i.d.\ random elements
  $( r^{-}_n, r^{+}_n,R^{-}_n, R^{+}_n)_{n\geq 1}$  with the same law
  of $(r^{-}, r^{+}, R^{-}, R^{+})$ defined in Section \ref{sec.general}.
\end{itemize}
Assume also that $(\ell_n)_{n \geq 1}$, $( R^{-}_n, R^{+}_n, r^{-}_n, r^{+}_n)_{n\geq 1}$ and $(X_j)_{j \geq 1}$
are stochastically independent.
Define recursively the random array $[\ww_{j,n},O_{j,n}]_{j=1,\dots,n+1,n \geq 0}$ by setting
\begin{align*}
  &O_{1,0}:=\eins_d, \ww_{1,0}:=1, \quad \text{and for all $n\ge1$}: \\
  &(O_{1,n},\ldots,O_{n+1,n})
  = \big(O_{1,n-1},\ldots,O_{\ell_n-1,n-1},O_{\ell_n,n-1}  R^{-}_n,O_{\ell_n,n-1}  R_n^+ ,O_{\ell_n+1,n-1},\ldots,O_{n,n-1}\big) \\
  &(\ww_{1,n},\ldots,\ww_{n+1,n})
  = \big(\ww_{1,n-1},\ldots,\ww_{\ell_n-1,n-1},\ww_{\ell_n,n-1}  r^{-}_n,\ww_{\ell_n,n-1}  r_n^+ ,\ww_{\ell_n+1,n-1},\ldots,\ww_{n,n-1}\big).
\end{align*}
This construction extends the one given in \cite{BaLa} for a class of one-dimensional generalized Kac equations. 
For given $n\ge1$, one should think of the quantities $\ww_{n,j}$ and $O_{n,j}$ as attached to the $n+1$ leaves
of a binary tree (whose shape is determined by $\ell_1$ to $\ell_n$) with $n$ internal nodes.
In the context of the Kac model, these binary trees are commonly referred to as \emph{McKean trees}.
\begin{proposition}
  \label{Prop.probrep00}
  For every $n \geq 0$, every $\rho\in\setR_+$ and every  $\rota\in\sod$ 
  one has
  \begin{align}
    \label{eq:pr}
    {\hat \mu}_n(\rho\rota\eed)=\E\Big[\prod_{j=1}^{n+1} {\hat \mu}_0\big(\rho\ww_{j,n} \rota O_{j,n} \eed\big)  \Big].
  \end{align}
\end{proposition}
\begin{proof} 
  For $n=0$ there is nothing to prove. 
  For $n=1$ the statement reduces to the definition of $Q_+$ in \eqref{eq:miracle}. 
  We proceed by induction on $n$. 

  Fix $n\geq 1$ and assume that \eqref{eq:pr} is true for all $k=0,\dots,n-1$ in place of $n$.
  By construction, $\ww_{1,1}=r^-_1$, $\ww_{1,2}=r^+_1$, and $O_{1,1}=R^-_1$, $O_{1,2}=R^+_1$.
  Consequently, we can write
  \begin{equation}\label{beta-O}
    \begin{split}
      & \ww_{j,n}=r^-_1\ww_{j,n}', \quad O_{j,n}=R^-_1 O_{j,n}' \qquad \text{for $j=1,\ldots,J$}, \\
      & \ww_{j,n}=r^+_1\ww_{j,n}'', \quad O_{j,n}=R^+_1 O_{j,n}'' \qquad \text{for $j=J+1,\ldots,n+1$},\\
    \end{split}
  \end{equation}
  with a random index $J\in\{1,\ldots,n\}$ depending on $\ell_1$ to $\ell_n$.
  The factorization \eqref{beta-O} corresponds to splitting the $n$th binary tree at the root
  into a left tree (with $J$ leaves) and a right tree (with $n+1-J$ leaves).
  It is easy to see that $J$ is uniformly distributed on $\{1,\ldots,n\}$, see e.g.\ \cite{BaLaMa}.
  It is further easy to see that, given $(J,r^-_1,r^+_1,R^-_1,R^+_1)$, 
  the random elements $(\ww_{j,n}',O_{j,n}')_{j=1,\ldots,J}$ and $(\ww_{j,n}'',O_{j,n}'')_{j=J+1,\ldots,n+1}$ are conditionally independent.
  Their conditional distribution, given the event $\{J=k\}$,
  satisfies
  \begin{align*}
    \big(\beta_{j,n}',O_{j,n}'\big)_{j=1,\dots,k} &\deq\big(\beta_{j,k-1},O_{j,k-1}\big)_{j=1,\dots,k}, \\
    \big(\beta_{j,n}'',O_{j,n}''\big)_{j=k+1,\dots,n+1}&\deq\big(\beta_{j,n-k},O_{j,n-k}\big)_{j=1,\dots,n+1-k}.
  \end{align*}
  Thus, if $(r^-,r^+,R^-,R^+)$ is defined as above and it is assumed independent of all the rest, using the induction hypothesis, one can write 
  \begin{align*}
    \E\Big[\prod_{j=1}^{n+1} {\hat \mu}_0\big(\rho\ww_{j,n} \rota O_{j,n} \eed\big)  \Big]
    &= \frac{1}{n} \sum_{k=1}^{n} \E\Big [\E\Big [ \prod_{j=1}^{k}  {\hat \mu}_0\big(\rho r^-\rota R^-\ww_{j,k-1}O_{j,k-1}\eed \big) \Big|r^-,r^+,R^-,R^+\Big] \\
    &\qquad \cdot \E\Big [\prod_{j=1}^{n+1-k}{\hat \mu}_0\big(\rho r^+\rota R^+\ww_{j,n-k}O_{j,n-k}\eed \big)\Big |r^-,r^+,R^-,R^+\Big] \Big ] \\
    &=\frac1n\sum_{k=1}^n \E\big[\hat\mu_{k-1}\big(\rho r^-\rota R^-\eed\big)\hat\mu_{n-k}\big(\rho r^+\rota R^+\eed\big)\big],
  \end{align*}
  which, by \eqref{eq:miracle} and \eqref{Wild2}, equals to $\hat\mu_n$.
\end{proof}
We can now formulate the above mentioned probabilistic representation. 
The first representation of this type has been derived in \cite{DoRe2} for the \emph{fully elastic} Boltzmann equation in $\setR^3$.
\begin{proposition}\label{probrep}
  Let $\ee$ in $\sphere$ and let $\CO\in\sod$ be such that $\ee=\CO\eed$.
  Then $\rho \mapsto {\hat \mu}_n(\rho \ee)$ is the characteristic function of
  \begin{align}
    \label{eq:pr1}
    \sum_{k=1}^{n+1}\big(\ww_{k,n}\CO O_{k,n}\eed\big) \cdot X_k.
  \end{align}
\end{proposition}
\begin{proof}
  We calculate the characteristic function of the sum given in \eqref{eq:pr1} at $\rho\in\setR_+$:
  \begin{align*}
    \E\bigg[\exp\bigg(i\rho\sum_{k=1}^{n+1}\big(\ww_{k,n}\CO O_{k,n}\eed\big) \cdot X_k\bigg)\bigg] 
    &=\E\bigg[\E\bigg[\prod_{k=1}^{n+1}\exp\Big(i\big(\rho\ww_{k,n}\CO O_{k,n}\eed\big)\cdot X_k\Big)\bigg|(\ww_{j,n},O_{j,n})_{1\le j\le n+1}\bigg]\bigg]\\
    &=\E\bigg[\prod_{k=1}^{n+1}\hat\mu_0\big(\rho\ww_{k,n}\CO O_{k,n}\eed\big)\bigg]
    =\hat\mu_n(\rho\CO\ee_d)=\hat\mu_n(\rho\ee),
  \end{align*}
  where we have used \eqref{eq:pr}.
  Since two characteristic functions that coincide on the positive real axis are equal, the first claim follows.
\end{proof}
\begin{remark}
  A consequence of Proposition \ref{probrep} is the following:
  if $(V_t)_{t\ge0}$ is a random process with (marginal) distribution $\mu(t)$ for every $t>0$
and $(N_t)_{t\ge0}$ is a random process with values in $\{0,1,\dots,\}$ and independent of $(\beta_{k,n}, O_{k,n})_{k,n}$ and $(X_k)_{k \geq 1}$,
  such that $\pby\{N_t=n\}=e^{-t}(1-e^{-t})^{n}$ for every $t\ge0$,
  then 
  \[
  (\rota\eed)\cdot V_t \deq \sum_{k=1}^{N_t+1} \big(\ww_{k,N_t}\rota O_{k,N_t}\eed\big)\cdot X_k
  \]
  for every $\rota\in\sod$ and all $t\ge0$.
  Indeed, it suffices to observe that
  \begin{align*}
    \E\big[\exp\big(i(\rota\eed)\cdot\rho V_t\big)\big]={\hat \mu}(\rho\rota\eed,t)
    =\sum_{n \geq 0} e^{-t}(1-e^{-t})^n{\hat \mu}_n(\rho \ee)
  \end{align*}
  by \eqref{eq:pr1} above and the Wild representation \eqref{Wild2zero}.
\end{remark}

\section{The inelastic Maxwell models as a special case}\label{sec.maxwell}
The aim of this section is to show that the homogeneous Boltzmann equation with collision rules \eqref{eq:maxwellrules} 
is indeed a special case of the more general equation considered here.
Theorem \ref{thm.simple} then follows as a corollary of Theorem \ref{thm.main}.

Our starting point is the equation in its Fourier representation \eqref{eq.boltz1},
which has been derived in \cite{Bobylev99}, with
the collision kernel 
\begin{equation}
  \label{eq.boltz2}
  \widehat{Q_+}[\hat\mu,\hat\mu](\xi) = \E[{\hat\mu}(Y_{\xi}^{+}){\hat \mu}(Y_{\xi}^{-})],
\end{equation}
where, for any $\xi\in\setR^d$, the two random vectors $Y_{\xi}^{-}$ and $Y_{\xi}^{+}$ in $\setR^d$ 
are given by
\begin{align}
  \label{eq:xipm2}
  &Y_{\xi}^{-} := \frac{1-\del}2(\xi - |\xi|\nml), \quad
  Y_{\xi}^{+} := \frac{1+\del}2\xi + \frac{1-\del}2|\xi|\nml, \\
  \label{eq:ximp3}
  &\text{with a random unit vector $\nml$ which has law $\gg\big(\sigma \cdot\xi/|\xi|\big)\dd\uuu(\sigma)$}.
\end{align}
Below, we rewrite \eqref{eq.boltz2} in the form \eqref{eq:miracle},
with suitable random quantities $r^\pm$ and $R^\pm$ satisfying (H1)-(H3).

\subsection{Preliminaries on rotation groups}
We start by recalling some well-known facts about the Haar distribution.
If the random matrix $O$ has Haar distribution on $\sodk$, 
then
\[
GO\deq O^T\deq O
\]
for every orthogonal matrix $G\in\sodk$; see, e.g., Theorem 5.14 in \cite{rudin}. 
Moreover, for any $\ee \in \mathbb{S}^{k-1}$, 
\begin{equation}\label{unif-on-sphere}
  O\ee \quad\text{is uniformly distributed on $\mathbb{S}^{k-1}$}.
\end{equation}
A random matrix $U$ with values in $\sod$ will be called \emph{uniformly distributed on $\soda$} 
if $U\in\soda$ a.s., and the random $(d-1)\times(d-1)$-matrix $U^*$ obtained from $U$ by deleting 
the $d$th line and $d$th column has as distribution the Haar measure of the $(d-1)$-dimensional rotation group.

We call a measure $\lambda$ on $\sphere$ \emph{invariant under $\soda$}, 
if $\lambda(\rota B)=\lambda(B)$ for every $\rota$ in $\soda$ and for all measurable sets $B\subseteq\sphere$.
Since $\soda$ acts transitively on each of the $(d-2)$-spheres $\{y\in\sphere|\eed\cdot y=z\}$ with $z\in(-1,1)$,
the invariant measure $\lambda$ is uniquely determined by its \emph{projected measure} $\Pi\lambda$ on $[-1,1]$,
given by $\Pi\lambda(J) = \lambda\big(\{y\in\sphere|\eed\cdot y\in J\}\big)$ for all measurable $J\subseteq[-1,1]$.
In the particular case that $\lambda(\dn\sigma) = f(\sigma\cdot\eed)\uuu(\dn\sigma)$  with $f:[-1,1]\to\setR$,
the projected measure $\Pi\lambda$ has a density (w.r.t.\ Lebesuge measure) $\Pi f:[-1,1]\to\setR$ with
\begin{align}
  \label{eq:project}
  \Pi f(z) := \frac{\dn(\Pi\lambda)}{\dd z} = \frac1{B_d}f(z)\sqrt{(1-z^2)^{d-3}}
  \quad \text{where} \quad
  B_d = \int_0^1 \sqrt{z^{-1}(1-z)^{d-3}}\dd z,
\end{align}
which is easily verified by the change of variables formula.

Finally, we denote by $Z_\psi\in\sod$  the (positive) rotation in the $\ee_1-\eed$-plane about the angle $\psi\in[0,\pi]$, 
that is
\begin{align*}
  (Z_\psi)_{kk}&=1,\qquad k=2,\dots,d-1,\\
  (Z_\psi)_{11}&=\cos\psi, \qquad (Z_\psi)_{1d}=-\sin\psi, \\
  (Z_\psi)_{d1}&=\sin\psi, \qquad (Z_\psi)_{dd}=\cos\psi,
\end{align*}
and all other entries of $Z_\psi$ are zero.
The following probabilistic interpretation of Hurwitz's \cite{hurwitz} representation of the Haar measure will be of importance.
\begin{theorem}
  \label{thm:hurwitz}
  There are random rotations $U_1,U_2$ in $\sod$ and a random angle $\psi_*$ in $[0,\pi]$ such that
  \begin{itemize}
  \item $U_1$, $U_2$ and $\psi_*$ are independent,
  \item $U_1$ is uniformly distributed on $\soda$ and  $U_2\in\soda$ a.s.,
  \item $\psi_*$ has a continuous probability density function that is positive on $(0,\pi)$,
  \item the law of $U_1Z_{\psi_*} U_2$ is the Haar measure on $\sod$.
  \end{itemize}
\end{theorem}
\begin{proof}[Sketch of the proof]
    In \cite{hurwitz} it is shown that 
    an arbitrary $d$-dimensional rotation matrix may be written as a product of $d(d-1)/2$ elementary rotations in two-dimensional subspaces. 
    Denote by $Z^{i,j}(\psi)$ the matrix of an elementary rotation in the plane $\ee_i-\ee_j$ of an angle $\psi$, i.e.  the only nonzero elements of $Z^{i,j}$ are
    \begin{align*}
      Z^{i,j}_{kk}(\psi)&=1,\qquad k=1,\dots,d,\qquad k\neq i,j\\
      Z^{i,j}_{ii}(\psi)&=\cos (\psi), \qquad \;\; Z^{i,j}_{ij}=\sin (\psi), \\
      Z^{i,j}_{ji}(\psi)&=-\sin (\psi), \qquad Z^{i,j}_{jj}=\cos (\psi).
    \end{align*}
    Then, any rotation matrix $\rot$ can be represented as 
    \[
    \rot =F_1 F_2 \dots F_{d-1}
    \]
    where
    \[
    F_i=Z^{d-i,d-i+1}(\psi_{i-1,i})Z^{d-i+1,d-i+2}(\psi_{i-2,i})\dots Z^{d-1,d}(\psi_{0,i}).
    \]
    The Haar distribution on $\sod$ is obtained if the generalized Euler angles
    $\psi_{j,i}$ are independent, $\psi_{0,i}$ are uniformly distributed on $[0,2\pi)$ for $i=1,\dots ,d-1$ and 
    $\psi_{r,s}$ for $r=1,\dots,s-1$ are absolutely continuous with density $\sin(\psi )^{r}\J_{[0,\pi)}(\psi)$.
    It is then easy to see that, as a consequence of the above representation, 
    one obtains the result.  
\end{proof}

\subsection{Definition of the probabilistic representation}
Given the cross section $\gg$ on $(-1,1)$, 
define the projected density $\Pi\gg$ according to \eqref{eq:project}.
Since $\gg$ is normalized as stated in \eqref{eq:ggnormal}, $\Pi\gg$ is a probability density.
Let $\psi$ be a random angle in $(0,\pi)$ such that $\cos\psi$ has $\Pi\gg$ as density,
which is equivalent to saying that $\psi$ itself is distributed with law
  \begin{align}
    \label{eq:dowereallyneedalltheseformulas}
    \gg(\cos\eta)\sin^{d-2}\eta\dd\eta.
  \end{align}
Further, let $U_1$, $U_2$ be random rotations in $\soda$ --- independent of each other and independent of $\psi$ --- with the properties from Theorem \ref{thm:hurwitz}.
In particular, $U_1$ is uniformly distributed on $\soda$.
From that, define two further random angles in $\psi^\pm$ in $(0,\pi)$ implicitly by
\begin{align}
  \label{angoli}
  \cos\psi^- = 2^{-1/2}\sqrt{1-\cos\psi},
  \quad
  \cos\psi^+ = 2^{-1/2}\frac{(1+\delta)+(1-\delta)\cos\psi}{\sqrt{(1+\delta^2)+(1-\delta^2)\cos\psi}}.
\end{align}
Now set
\begin{align}
  \label{eq:defrRm}
    r^- &:= 2^{-1/2}(1-\delta)\sqrt{1-\cos\psi}, & R^- := U_1Z_{\psi^-}U_2, \\
  \label{eq:defrRp}
    r^+ &:= 2^{-1/2}\sqrt{(1+\delta^2)+(1-\delta^2)\cos\psi}, & R^+ := U_1Z_{\psi^+}U_2.
\end{align}
\begin{proposition}\label{lemmaProbrepK}
  For every vector $\xi$ and every $\rot\in\sod$ such that $\xi=|\xi|\rot\eed$ one has
  \begin{equation}\label{eq.ypm}
    \big( Y_{\xi}^{-},  Y_{\xi}^{+} \big) \deq \big(|\xi| r^- \rot R^- \eed, |\xi| r^+ \rot R^+ \eed).
  \end{equation}
\end{proposition}
The essential ingredient of the proof is the following.
\begin{lemma}
  \label{lem:nml}
  For $\xi=\eed$, 
  the random unit vector $\nml$ in \eqref{eq:xipm2} admits the representation $\nml\deq U_1Z_\psi U_2\eed$.
\end{lemma}
\begin{proof}
   We need to show that the law $\lambda$ of $U_1Z_\psi U_2\eed$ is the same as the law $\lambda'$ of $\nml$.
  Both $\lambda$ and $\lambda'$ are invariant under $\soda$:
  for $\lambda'$, this is clear by definition in \eqref{eq:ximp3}.
  For $\lambda$, this follows since $U_1$, $\psi$ and $U_2$ are independent, 
  and $GU_1\deq U_1$ for every $G\in\soda$.
  By our considerations on $\soda$-invariant measures above, 
  it therefore suffices to show that the projected measures are equal, $\Pi\lambda=\Pi\lambda'$.
  
  For $\lambda'$, we obtain from the definition of $\nml$ and formula \eqref{eq:project} that $\nml\cdot\eed$ has law $\Pi\gg$.
  Concerning $\lambda$, recall that $U_1$ and $U_2$ take values in $\soda$ a.s.,
  which implies that 
  \begin{align*}
    \eed\cdot U_1Z_\psi U_2\eed = (U_1^T\eed)\cdot(Z_\psi U_2\eed) = \eed\cdot(Z_\psi\eed) = \cos\psi,
  \end{align*}
  using the definition of $Z_\psi$.
  The claim now follows since $\cos\psi$ has law $\Pi\gg$ by definition. 
 \end{proof}
\begin{proof}[Proof of Proposition \ref{lemmaProbrepK}]
Let $\xi=|\xi|\rota\eed$ be given. For any bounded continuos function $f$
\[
 \begin{split}
 \E[f(Y_{\xi}^{-})]
&=\int_{\sphere} f\Big(\frac{1-\del}2 |\xi| \CO ( \eed- \CO^T \s)\Big) \gg(\CO ^T \s \cdot  \eed) \uuu(d\s)\\
&=\int_{\sphere} f\Big(|\xi| \CO \frac{1-\del}2   ( \eed- \s)\Big) \gg(\s \cdot \eed) \uuu(d\s)
=  \E[f(|\xi| \CO Y_{\eed}^{-})], \\
\end{split}
\]
where we have used \eqref{eq:xipm2}--\eqref{eq:ximp3} and a change of variables in the integral.
Hence $ Y_\xi^-  \deq |\xi|\rot Y_{\eed}^-$.
  Since $Y_\xi^-+Y_\xi^+=\xi$ and $Y_{\eed}^++Y_{\eed}^-=\eed$,
  it follows further that
  \begin{align*}
    (Y_{\xi}^{-},Y_{\xi}^{+}) \deq (|\xi| \rot Y_{\eed}^{-},|\xi| \rot Y_{\eed}^{+}).
  \end{align*}
  It is thus sufficient to prove the claim for $\xi=\eed$ and $\rot=\eins_d$.
  By Lemma \ref{lem:nml}, we have
  \begin{align*}
    (Y_{\eed}^{-},Y_{\eed}^{+})
    &\deq \bigg(\frac{1-\delta}2(\eed-U_1Z_\psi U_2\eed),\frac{1+\delta}2\eed+\frac{1-\delta}2U_1Z_\psi U_2\eed\bigg) \\
    &= \bigg( U_1\bigg[\frac{1-\delta}2(\eins_d-Z_\psi)\bigg]U_2\eed,U_1\bigg[\frac{1+\delta}2\eins_d+\frac{1-\delta}2Z_\psi\bigg]U_2\eed\bigg),
  \end{align*}
  To finish the proof, observe that 
  we have
  \begin{align*}
    U_1\bigg[\frac{1-\delta}2(\eins_d-Z_\psi)\bigg]U_2\eed = r^-R^-\eed,
    \quad
    U_1\bigg[\frac{1+\delta}2\eins_d+\frac{1-\delta}2Z_\psi\bigg]U_2\eed = r^+R^+\eed,
  \end{align*}
  which easily follows from our definitions of $r^\pm$ and $R^\pm$ by elementary geometric considerations.
\end{proof}

\subsection{Verification of (H1)--(H3)}
It remains to verify that the random quantities defined in \eqref{eq:defrRm}-\eqref{eq:defrRp} satisfy the hypotheses (H1)--(H3).
Condition (H1) is a direct consequence of Proposition \ref{lemmaProbrepK},
since with $\xi:=\rota_1\eed=\rota_2\eed$, one has that
\begin{align*}
  \big(\rota_1r^-R^-\eed,\rota_1r^+R^+\eed) \deq (Y^-_\xi,Y^+_\xi) 
  \deq\big(\rota_2r^-R^-\eed,\rota_2r^+R^+\eed).
\end{align*}
The validity of (H2) is a consequence of the following.
\begin{lemma}\label{lemmainelastic2}
  There is a unique $\alpha \in (0,2)$ such that $\E[(r^+)^{\alpha}+(r^-)^{\alpha}]=1$,
  and $\E[(r^+)^{\alpha\gamma}+(r^-)^{\alpha\gamma}]<1$ for every $\gamma>1$.
\end{lemma}
\begin{proof}
  Recall the convex function $\CS$ defined in \eqref{eq:CS}.
  We have
  \[
  \CS(s) = \E\Big[\Big((1-\del)^2\frac{1-\cos\psi}2\Big)^{s/2}\Big] 
  + \E\Big[\Big(\frac{1+\del^2}2+\frac{1-\del^2}2\cos\psi\Big)^{s/2}\Big]-1.
  \]
  On one hand, $\CS(0)=1$, because $\psi$ is an absolutely continuous random variable.
  On the other hand, since $0<r^\pm<1$ almost surely,
  it follows that $\lim_{s \to +\infty} \CS(s)=-1$.
  Finally, at $s=2$, we have
  \[
  \CS(2) = \E\big[(r^+)^2+(r^-)^2\big]-1
  = \delta(\delta-1) \E [1-\cos\psi] <0.
  \]
  By convexity of $\CS$, this proves the claim.
\end{proof}
For the verification of assumption (H3), we employ the Hurwitz' representation of the Haar measure
given in Theorem \ref{thm:hurwitz}.
\begin{lemma}
  \label{haar-a.c.} 
  The probability measures $\pb^\pm$ defined in \eqref{eq:BB} are absolutely continuous with respect to the Haar measure.
\end{lemma}
\begin{proof}
  Recall Theorem \ref{thm:hurwitz}, and let $U_1,U_2$ and $\psi_*,\psi$ be chosen as indicated above.
  Further, observe that, since $U_1,U_2,\psi$ are independent, 
  and since the law of $\psi$ is given in \eqref{eq:dowereallyneedalltheseformulas},
  one can write, for every $f\in C_0^b(\sod)$,
  \begin{equation*}
    \begin{split}
      \int_\sod f(R) \pb^\pm(\dn R) 
      &  = \frac{\E\big[(r^\pm)^\alpha f\big(U_1Z_{\psi^\pm }U_2\big)\big]}{\E[(r^\pm)^\alpha]} \\
      & =\frac{\E\big[ \int_{(0,\pi)} (r^\pm(\eta))^\alpha f\big( U_1Z_{\psi^\pm(\eta)} U_2\big) \gg(\cos\eta)\sin^{d-2}\eta\dd\eta \big] }{\int_{(0,\pi)} (r^\pm(\eta))^\alpha \gg(\cos\eta)\sin^{d-2}\eta\dd\eta } 
    \end{split}
  \end{equation*}
  where $\psi^\pm(\eta)$ and $r^\pm(\eta)$ are defined as functions of $\eta$ via \eqref{angoli}--\eqref{eq:defrRp} using $\eta$ in place $\psi$. 
  Hence
  \begin{equation*}
    \int_\sod f(R) \pb^\pm(\dn R) = \E\big[ f\big( U_1Z_{\tilde\psi^\pm}U_2\big)\big],
  \end{equation*}
  where $\tilde\psi^\pm$ are defined via \eqref{angoli} from a random angle $\tilde\psi$ 
  --- being independent of $U_1$ and $U_2$ ---
  in $(0,\pi)$ with law
  \[
  \frac{ \big(r^\pm(\eta)\big)^\alpha \gg(\cos\eta)\sin^{d-2}\eta\dd\eta}{\int_{(0,\pi)} (r^\pm(u))^\alpha \gg(\cos u)\sin^{d-2}u\dd u}.
  \] 
  It thus suffices to show that the laws of the random rotations $U_1Z_{\tilde\psi^\pm}U_2$ are absolutely continuous with respect to the law of $U_1Z_{\psi_*}U_2$.
  Since $\tilde\psi$ has a density on $(0,\pi)$, also $\cos\tilde\psi^\pm$ given via \eqref{angoli} have densities on $(-1,1)$,
  and thus $\tilde\psi^\pm$ themselves have densities on $(0,\pi)$, all with respect to the Lebesgue measure on the respective intervals.
  Since further the density of $\psi_*$ is positive on $(0,\pi)$, it follows that the laws of $\tilde\psi^\pm$ are absolutely continuous with respect to that of $\psi_*$.
  Then also the law of the triple $(U_1,\tilde\psi_\pm,U_2)$  is absolutely continuous with respect to the law of $(U_1,\psi_*,U_2)$ on $\soda\times(0,\pi)\times\soda$.
  And the respective images in $\sod$ under the continuous map $(G_1,\theta,G_2)\mapsto G_1Z_\theta G_2$ inherit the absolute continuity.
\end{proof}

\section{Study of an instrumental process on $\czero$}\label{S:walk}
This section is devoted to the proof of convergence 
of the following auxiliary random processes $(\Psi_n)_{n\ge0}$ taking values in $\czero$.
Given a continuous function $\Psi_0\in\czero$,
define for all $n\ge1$:
\begin{equation}
  \label{defPsin}
  \Psi_n(\rota) :=\sum_{j=1}^{n+1}  \ww_{j,n}^\alpha\Psi_0(\rota O_{j,n} ).
\end{equation}
Throughout this section, we continue to assume hypothesis (H1)--(H3).

The ultimate goal is to show convergence of $\Psi_n$ to a (random) constant function 
in the sense made precise in Proposition \ref{convPsin} below.
In order to characterize the limit, we start with an auxiliary result.
\begin{lemma}
  \label{Lemmaweights}
  The random quantities 
  \begin{align}
    \label{eq.them}
    M_n^{(\alpha)} := \sum_{j=1}^{n+1} \beta_{j,n}^\alpha  \quad \text{and} \quad
    \beta_{(n)}:=\max_{j=1,\dots,n+1} \beta_{j,n},
  \end{align}
  have the following properties:
  \begin{itemize}
  \item[(i)] $\E[ M_n^{(\alpha)}]=1$ for every $n$.
  \item[(ii)] $M_n^{(\alpha)}$ converges almost surely to a random variable $M_\infty^{(\alpha)}$ as $n\to +\infty$.
    The characteristic function of $M_\infty^{(\alpha)}$ satisfies equation \eqref{fixedpointDL} and $\E[(M_\infty^{(\alpha)})]=1$.
    Moreover, $\E[(M_\infty^{(\alpha)})^p]<+\infty$ if and only if $\E[(r^+)^{\a p}+(r^-)^{\a p}]<1$.
  \item[(iii)] $\beta_{(n)}$ converges to zero in probability.
  \end{itemize}
\end{lemma}
\begin{proof}
Claims (i) and (ii) are contained in Proposition 2 of \cite{BaLaMa},
    while claim (iii) is Lemma 3 in \cite{BaLaMa}. 
\end{proof}
The main result of this section is:
\begin{proposition}\label{convPsin}
  $\Psi_n$ converges weakly to $m_0M_\infty^{(\alpha)}$ in $\czero$,
  where
  \begin{align}
    \label{eq:defm0}
    m_0:=\int_{\sod} \Psi_0(\rota) \haar(\dn\rota).
  \end{align}
  Hence,  for every $\rota\in\sod$, 
  the sums $\sum_{j=1}^{n+1} \ww_{j,n}^\alpha\Psi_0(\rota O_{j,n} )$ converge weakly to $m_0 M_\infty^{(\alpha)}$.
\end{proposition}
For the sake of simplicity the proof of Proposition \ref{convPsin} is split into several steps. 
Some of them use techniques developed in \cite{BaMa}.

\subsection{Basic properties of $\Psi_n$.}
Introduce the $L^p$-norms with respect to the Haar measure $\haar$ 
on measurable functions $f:\sod\to\setR$ as usual:
\begin{align*}
  \|f\|_{L^p} := \Bigg(\int_{\sod}|f(\rota)|^p\haar(\dn\rota)\Bigg)^{1/p} \quad \text{for all $p\ge1$},
  \quad
  \|f\|_{L^\infty} := \operatorname{ess\,sup}_{\rota\in\sod}|f(\rota)|.
\end{align*}
\begin{lemma}\label{lemma15}
  For every $n \geq 0$,
  \begin{equation}
    \label{eq:m0const}
    \E\bigg[ \int_{\sod} \Psi_n(\rota)\haar(\dn\rota) \bigg] =  m_0,
  \end{equation}
  where $m_0$ is given in \eqref{eq:defm0}, and
  \begin{equation}
    \label{boundL2normPsi}
    \E\big[ \|\Psi_n \|_{L^2} \big]\leq \|\Psi_0\|_{L^\infty}.
  \end{equation}
\end{lemma}
\begin{proof}
  Since $\haar$ is right invariant,
  \[
  \int_{\sod} \Psi_n(\rota)\haar(\dn\rota)
  = \int_{\sod} \sum_{j=1}^{n+1}\beta_{j,n}^\alpha \Psi_0(\rota O_{j,n})\haar(\dn\rota)
  = \sum_{j=1}^{n+1}\beta_{j,n}^\alpha \int_{\sod} \Psi_0(\rota)\haar(\dn\rota).
  \]
  Now \eqref{eq:m0const} follows by means of (i) in Lemma \ref{Lemmaweights}.
  Another application of that property yields \eqref{boundL2normPsi}:
  \begin{align*}
    \E\big[ \|\Psi_n \|_{L^2}\big] &
    = \E\Big[ \Big ( \int_{\sod}  \Big ( \sum_{j=1}^{n+1}  \ww_{j,n}^\alpha \Psi_0(\rota O_{j,n})  \Big)^{2} \haar(\dn\rota)\Big)^{1/2}\Big]
    \leq \|\Psi_0\|_{L^\infty}\E\Big[  \sum_{j=1}^{n+1}  \ww_{j,n}^\alpha \Big]= \|\Psi_0\|_{L^\infty}.
  \end{align*}
\end{proof}
\begin{lemma}\label{tightnessPsin} 
The  laws of  $\Psi_n$ form a tight sequence of probability measures  on $\czero$ and hence they are relatively sequentially compact.
\end{lemma}
\begin{proof}
  By the classical tightness criterion for sequences of random continuous functions, see e.g.\ Theorem 16.5 in \cite{Kallenberg},
  it suffices to show that
  \[
  w(\Psi_n,\delta):=\sup\big\{|\Psi_n(\rota_1)-\Psi_n(\rota_2)|\,\big|\, \|\rota_1-\rota_2\|_* \leq \delta \big\} ,
  \]
  where $\|\cdot\|_*$ is the matrix (operator) norm induced by the euclidean norm on $\setR^d$,
  satisfies
  \begin{align}
    \label{eq:kallenberg}
    \lim_{\delta\to 0}\limsup_{n\to \infty}\E[ w(\Psi_n,\delta) ] =0.
  \end{align}
  Observe that for arbitrary $\rota_1,\rota_2\in\sod$,
  \begin{align*}
    |\Psi_n(\rota_1)-\Psi_n(\rota_2)|
    =\bigg|\sum_{j=1}^{n+1}\beta_{j,n}^\alpha [\Psi_0(\rota_1 O_{j,n})-\Psi_0(\rota_2 O_{j,n})]\bigg|
    \leq \sum_{j=1}^{n+1}\beta_{j,n}^\alpha |\Psi_0(\rota_1')-\Psi_0(\rota_2')|,
  \end{align*}
  with $\rota_i'=\rota_i O_{j,n}$ for $i=1,2$.
  Since $O_{j,n}$ is a rotation matrix,
  \begin{align*}
    \|\rota_1'-\rota_2'\|_* = \| (\rota_1-\rota_2)O_{j,n} \|_*  = \|\rota_1-\rota_2 \|_*.    
  \end{align*}
  It follows that
  \begin{align*}
    \E[ w(\Psi_n,\delta) ] 
    \le \E\bigg[\sum_{j=1}^{n+1}\beta_{j,n}^\alpha\bigg] 
    \sup\big\{|\Psi_0(\rota_1)-\Psi_0(\rota_2)|\,\big|\, \|\rota_1-\rota_2\|_* \leq \delta \big\} .
  \end{align*}
  The expectation value on the right-hand side equals to one, independently of $n$, by Lemma \ref{Lemmaweights} (i).
  The supremum, which is also independent of $n$, tends to zero for $\delta\downarrow0$,
  since the continuous function $\Psi_0$ on the compact manifold $\sod$ is automatically \emph{uniformly} continuous.
  Since $\czero$ is a Polish space, the last part of the statement follows from Prohorov's Theorem, see e.g. Thm. 17, Chapter 18 in \cite{fristed}. 
\end{proof}

\subsection{Definition of the recursion operator $T$}
Given $A\in\sod$ and a function $f$ on $\sod$, we denote by $A^\#f$ and $A_\#f$ the functions given by
\begin{align}
  \label{eq:8}
  A^\# f(\rota) = f(\rota A) \quad \text{and} \quad A_\#f(\rota) = f(\rota A^T)
  \quad \text{for all $\rota\in\sod$}.
\end{align}
Observe that $A^\#(B^\#f)=(AB)^\#f$ for arbitrary $A,B\in\rota$, since
\begin{align*}
  A^\#(B^\#f)(\rota) = B^\#f(\rota A) = f(\rota AB).
\end{align*}
With these notations,
\[
\Psi_n(\rota)=\sum_{j=1}^{n+1}\beta_{j,n}^\alpha \Psi_0(\rota O_{j,n})=\sum_{j=1}^{n+1}\beta_{j,n}^\alpha O_{j,n}^\#\Psi_0(\rota).
\]
Introduce a sequence $(\nu_n)$ of probability measures on $\czero$ by
\begin{equation}
  \label{eq:22}
  \nu_0 :=\delta_{\Psi_0}, \quad\text{and for every $n\ge1$,}\quad
  \nu_n := \law \Big (  \sum_{j=1}^{n+1}\beta_{j,n}^\alpha O_{j,n}^\#\Psi_0 \Big ).
\end{equation}
Next, define a recursion operator $T$ 
on the set $\CP(\czero)$ of all probability measures on $\czero$ as follows.
Given $\nu',\nu''\in\CP(\czero)$, let $\Psi'$ and $\Psi''$ be two independent random functions with distributions $\nu'$ and $\nu''$, respectively,
which are also independent of $(r^-,r^+,R^-,R^+)$.
Then define
\begin{equation}
  \label{eq:6}
  T[\nu',\nu''] := \law \big( (r^-)^\alpha (R^-)^\#\Psi' + (r^+)^\alpha (R^+)^\#\Psi'' \big).
\end{equation}
$T$ has a fixed point:
set $\Psi_\infty:=m_0 M_\infty^{(\alpha)}$ and \(\nu_\infty:=\text{Law}(\Psi_\infty)\).
Using Lemma \ref{Lemmaweights}, it is easy to see that
\begin{equation}
  \label{nuinftyfixT}
  \nu_\infty=T(\nu_\infty,\nu_\infty).
\end{equation}
In the following, we shall show that this fixed point is attractive in a suitable metric.
\begin{lemma} 
  For each $n\geq 1$, the following recursion relation holds:
  \begin{align}
    \label{eq:10}
    \nu_n = \frac1n \sum_{k=1}^n T[\nu_{k-1},\nu_{n-k}].
  \end{align}
\end{lemma}
\begin{proof} The proof is similar to the one of Proposition \ref{Prop.probrep00}.
  With the notations \eqref{beta-O}, we can write
  \begin{align*}
    \Psi_n
    &= \sum_{j=1}^{J}\big(r^-_1\beta_{j,n}'\big)^\alpha (R^-_1O_{j,n}')^\#\Psi_0
    + \sum_{j=J+1}^{n+1}\big(r^+_1\beta_{j,n}''\big)^\alpha (R^+_1O_{j,n}'')^\#\Psi_0 \\
    &= (r^-_1)^\alpha (R^-_1)^\#\Big(\sum_{j=1}^{J}(\beta_{j,n}')^\alpha (O_{j,n}')^\#\Psi_0\Big)
    + (r^+_1)^\alpha (R^+_1)^\#\Big(\sum_{j=J+1}^{n+1}(\beta_{j,n}'')^\alpha (O_{j,n}'')^\#\Psi_0\Big),
  \end{align*}
  using the rule $A^\#(B^\#f)=(AB)^\#f$ discussed above.
  To conclude, observe that --- conditionally on $\{J=k\}$ ---
  \[
    \sum_{j=1}^{J}(\beta_{j,n}')^\alpha (O_{j,n}')^\#\Psi_0 \deq \Psi_{k-1}, \qquad
    \sum_{j=J+1}^{n+1}(\beta_{j,n}'')^\alpha (O_{j,n}'')^\#\Psi_0 \deq \Psi_{n-k}.
   \qedhere
    \]
\end{proof}
The goal for the rest of this section is to show that the map $T$ is a contractive in an appropriate metric.
  Once this is shown, the proof of Propostion \ref{convPsin} follows easily.

\subsection{Contraction in Fourier distance}
Recall that $L^2(\sod,\haar)$ is a real Hilbert space with respect to the scalar product
\[
\spr{g}{f}=\int_\sod g(\rot)f(\rot)\dd\haar(\rot).
\]
For a probability measure $\nu$ on $\czero$, define its $L^2$-characteristic functional (or Fourier transform) $\hat\nu:L^2(\sod,\haar)\to\setC$ by
\begin{align*}
  \hat\nu(g) := \int_{\czero} \exp\big(i\spr{g}{f}\big)\, \nu(\dn f)=\E\big[\exp\big(i\spr{g}{\Psi}\big)\big],
\end{align*}
where $\Psi$ is a random function with law $\nu$.

Now, given $\gamma>1$, introduce the \emph{Fourier distance} between any $\nu',\nu''\in\CP(\czero)$ by
\begin{align*}
  d_\gamma(\nu',\nu'')
  &:= \sup_{0\neq g\in L^2}\frac{|\E\big[\exp(i\spr{g}{\Psi'})-\exp(i\spr{g}{\Psi''})-i\spr{g}{\Psi'-\Psi''}\big]|}{\|g\|_{L^2}^\gamma} \\
  &= \sup_{0\neq g\in L^2}\frac{\big|\hat\nu'(g)-\hat\nu''(g)-i\spr{g}{\Delta}\big|}{\|g\|_{L^2}^\gamma}
  \quad\text{with}\quad \Delta:=\E[\Psi'-\Psi''],
\end{align*}
where $\Psi'$ and $\Psi''$ are two random functions distributed according to $\nu'$ and $\nu''$.
This is a variant of the original Fourier metric, 
which was first introduced in the context of kinetic equations in \cite{Wennberg}, 
and has since then been generalized in manifold ways, see e.g.\ \cite{BaMa} for another application to measures on matrices.
Strictly speaking, this distance is not a metric since it might attain the value infinity.
Notice further that $\nu'$ and $\nu''$ might be ``close'' with respect to $d_\gamma$ even if their expectation values differ significantly. 
\begin{lemma}\label{lemmacontraction}  
  Let $\nu_1',\nu_2'$ and $\nu_1'',\nu_2''$ probability measures on $\czero$
  and $\Psi_1',\Psi_1'',\Psi_2',\Psi_2''$ be independent random functions
  with laws $\nu_1',\nu_2'$ and $\nu_1'',\nu_2''$, respectively.
  Then, for a given $\gamma \in (1,2)$,
  \begin{align*}
    &d_\gamma\big(T[\nu_1',\nu_1''],T[\nu_2',\nu_2'']\big)
    \le \E[(r^-)^{\alpha\gamma}]d_\gamma(\nu_1',\nu_2') + \E[(r^+)^{\alpha\gamma}]d_\gamma(\nu_1'',\nu_2'') \\
    &+ \max\big(2,\E[\|\Psi_1''\|_{L^2}],\E[\|\Psi_2'\|_{L^2}]) 
    \E[(r^-)^{\alpha\gamma}+(r^+)^{\alpha\gamma}]\big(\|\E[\Psi'_1-\Psi'_2]\|_{L^2}+\|\E[\Psi''_1-\Psi''_2]\|_{L^2}\big).
  \end{align*}
\end{lemma}
\begin{proof} 
  Let $\Psi_1',\Psi_1'',\Psi_2',\Psi_2''$ and $(r^-,r^+,R^-,R^+)$ be independent.
  We proceed in full analogy to the proof of Lemma 6 in  \cite{BaMa}.
  Set $\nu_j:=T[\nu_j',\nu_j'']$ and let $\Psi_j$ be distributed with laws $\nu_j$, respectively.
  Recalling the definition of $R^\#$ and $R_\#$ from \eqref{eq:8}, we obtain
  \begin{align*}
    \hat \nu_j(g) 
    &= \E\big[\exp\big(i\spr{g}{(r^-)^\alpha(R^-)^\#\Psi_j'+(r^+)^\alpha(R^+)^\#\Psi_j''}\big)\big] \\
    &= \E\big[\hat\nu_j'\big((r^-)^\alpha(R^-)_\# g\big) \hat\nu_j''\big((r^+)^\alpha(R^+)_\#g\big)\big] \\
    &= \E\big[\hat\nu_j'(g')\hat\nu_j''(g'')\big],
  \end{align*}
  with $g' := (r^-)^\alpha(R^-)_\#g$ and $ g'' := (r^+)^\alpha(R^+)_\#g$.
  Next, define $ \Delta':=\E[\Psi_1'-\Psi_2']$, $\Delta'':=\E[\Psi_1''-\Psi_2'']$, and observe that
  \begin{align*}
    \Delta := \E[\Psi_1-\Psi_2] = \E[(r^-)^\alpha(R^-)^\star\Delta' + (r^+)^\alpha(R^+)^\star\Delta''],
  \end{align*}
  so that $\spr{g}{\Delta} = \E\big[\spr{g'}{\Delta'}+\spr{g''}{\Delta''}\big].$
  We thus have
  \begin{align*}
    \delta(g) &:=\hat\nu_1(g)-\hat\nu_2(g)-i\spr{g}{\Delta} \\
    &= \E\big[\hat \nu_1'(g')\hat \nu_1''(g'')-\hat \nu_2'(g')\hat \nu_2''(g'')-i\spr{g'}{\Delta'}-i\spr{g''}{\Delta''}\big] \\
    &= \E\big[\big(\hat \nu_1'(g')-\hat \nu_2'(g')-i\spr{g'}{\Delta'}\big)\hat \nu_1''(g'')-i\big(1-\hat \nu_1''(g'')\big)\spr{g'}{\Delta'}\big] \\
    &\qquad + \E\big[\hat \nu_2'(g')\big(\hat \nu_1''(g'')-\hat \nu_2''(g'')-i\spr{g''}{\Delta''}\big)-i\big(1-\hat \nu_2'(g')\big)\spr{g''}{\Delta''}\big].
  \end{align*}
  Hence,
  we find
  \begin{align*}
    |\delta(g)|
    &\le \E\big[\|g'\|_{L^2}^\gamma\big]\,d_\gamma( \nu_1',\nu_2')
    + \E\big[\|g''\|_{L^2}^\gamma\big]\,d_\gamma(\nu_1'',\nu_2'') \\
    &\qquad + \E\big[\big|1-\hat \nu_1''(g'')\big|\|g'\|_{L^2}\big]\|\Delta'\|_{L^2}
    + \E\big[\big|1-\hat \nu_2'(g')\big|\|g''\|_{L^2}\big]\|\Delta''\|_{L^2}.
  \end{align*}
  Since $R^-$ and $R^+$ are orthogonal matrices, we have
  \begin{align*}
    \|g'\|_{L^2} = (r^-)^\alpha\|g\|_{L^2}, \quad
    \|g''\|_{L^2} = (r^+)^\alpha\|g\|_{L^2},
  \end{align*}
  and so
  \begin{align*}
    \|g\|_{L^2}^{-\gamma}|\delta(g)|
    &\le \E\big[(r^-)^{\alpha\gamma}\big]\,d_\gamma(\nu_1',\nu_2') + \E\big[(r^+)^{\alpha\gamma}\big]\,d_\gamma(\nu_1'',\nu_2'') \\
    & \quad +  \E\Big[\frac{|1-\hat \nu_1''(g'')|}{\|g\|_{L^2}^{\gamma-1}}(r^-)^\alpha\Big]\|\Delta'\|_{L^2}
    + \E\Big[\frac{|1-\hat \nu_2'(g')|}{\|g\|_{L^2}^{\gamma-1}}(r^+)^\alpha\Big]\|\Delta''\|_{L^2} \\
    &\le \E\big[(r^-)^{\alpha\gamma}\big]\,d_\gamma(\nu_1',\nu_2') + \E\big[(r^+)^{\alpha\gamma}\big]\,d_\gamma(\nu_1'',\nu_2'') \\
    & \quad +  \E\Big[(r^-)^\alpha(r^+)^{\alpha(\gamma-1)}\Big]\sup_{h\neq 0}\Big(\frac{|1-\hat \nu_1''(h)|}{\|h\|_{L^2}^{\gamma-1}}\Big) \|\Delta'\|_{L^2} \\
    & \quad + \E\Big[(r^-)^{\alpha(\gamma-1)}(r^+)^\alpha\Big]\sup_{h\neq 0}\Big(\frac{|1-\hat \nu_2'(h)|}{\|h\|_{L^2}^{\gamma-1}}\Big)\|\Delta''\|_{L^2}.
  \end{align*}
  By definition of the Fourier transform, and since $|1-e^{ix}|\leq |x|$, it follows that
  \[
  \big|1-\hat \nu_1''(h)\big| \leq \|h\|_{L^2}\E\big[\|\Psi_1''\|_{L^2}\big].
  \]
  Since also $|1-\hat\nu_1''(h)|\leq 2$ for all $h$,
  we have that
  \begin{align*}
    \sup_{h\neq 0}\Big(\frac{|1-\nu_1''(h)|}{\|h\|_{L^2}^{\gamma-1}}\Big) \leq \max\{2, \E\big[\|\Psi_1''\|_{L^2}\big]\},
  \end{align*}
  and similarly for the other supremum.
  To finish the proof, observe that by Young's inequality
  \[
  \E\big[(r^-)^\alpha(r^+)^{\alpha(\gamma-1)} + (r^-)^{\alpha(\gamma-1)}(r^+)^\alpha\big]
  \le \E[(r^-)^{\alpha\gamma}+(r^+)^{\alpha\gamma}].
  \qedhere
  \]
\end{proof}

\subsection{Contraction of means}
Lemma \ref{lemmacontraction} almost yields contractivity of $T$ in the Fourier distance $d_\gamma$ for some appropriate $\gamma>1$.
  Below, we provide a control on the remainder term, given by the $L^2$-distance of the expectation values of the argument measures.
\begin{proposition}\label{propD}  
There are constants $\kappa^- < \E[(r^{-})^{\alpha}]$ and $\kappa^+  < \E[(r^{+})^{\alpha}]$ such that 
\[
\|\E[(r^\pm)^\alpha (R^\pm)^\# f \|_{L^2} \leq \kappa^\pm \|f \|_{L^2} \quad
\text{for every $f\in\czero$ with}\quad \int_{\sod}f(\rota)\haar(\dn\rota)=0.
\]
\end{proposition}
To prove Proposition \ref{propD} we need some preliminary results.
Recalling the definition of $\pb^\pm$ from \eqref{eq:BB},
introduce continuous linear operators $L^\pm$ on $L^2(\sod , \haar)$ by
\begin{align*}
  (L^\pm f)(\rot) := \int_\sod f(\rot R) \pb^\pm(\dn R).
\end{align*}
Since for every $f,g\in L^2(\sod, \haar)$, we have that
 \begin{align*}
    \langle L^\pm f,g\rangle_{L^2}
    = \int_{\sod^2} f(\rot R)g(\rot) \pb^\pm(\dn R)\haar(\dn \rot) 
    = \int_{\sod^2 } f(\rot')g(\rot' R^T)\pb^\pm(\dn R)\haar(\dn\rot'),
  \end{align*}
it follows that the adjoint operator $(L^\pm)^* $ of $L^\pm$ is given by
\begin{align*}
  ((L^\pm)^* f)(\rot) = \int_\sod f(\rot R^T)\pb^\pm(\dn R).
\end{align*}
Consider the symmetric operator  $(L^\pm)^* L^\pm$ on $L^2(\sod , \haar)$,
which can be written as
\begin{align*}
  ((L^\pm)^* L^\pm f)(\rot) = \int_{\sod^2} f(\rot R_2^T R_1) \pb^\pm(\dn R_1)\pb^\pm(dR_2)
  = \int_{\sod} f(\rot B)\tpb^\pm(\dn B),
\end{align*}
where we define $\tpb^\pm$ as the law of the random rotation $R_2^TR_1$ for independent $R_1$, $R_2$ with distribution $\pb^\pm$ each.
It is easy to see that the powers of $(L^\pm)^* L^\pm$ admit the representations
\[
 [(L^\pm)^* L^\pm] ^n f(\rot) =\int_\sod f\big( \rot B\big) (\tpb^\pm)^{\star n}(\dn B), 
\]
where $\,^{\star n}$ denotes the n-fold convolution of a measure.
The following result is essential for the proof of Proposition \ref{propD}.
\begin{proposition}[Bhattacharya]\label{Bhatta2}
  Let  $G$  be  a  compact,  connected,  Hausdorff  group and let $\beta$ be a probability measure on $G$ such that
  $\beta$ has a nonzero absolutely continuous component with respect to the normalized Haar measure $\haar$ on $G$.
  Then there is $n\geq 1$  and $0<c\leq 1$ such that
  \begin{align}
    \label{eq:bhatta}
    \beta^{\star 2^n}(B) \geq c \haar(B)
  \end{align}
  for every measurable $B \subset G$.
\end{proposition}
Actually, in the proof of Theorem 3 in \cite{Bhattacharya} it is shown that there are 
a set $A\subseteq G$ of positive Haar measure, a positive number $\bar c>0$ and an index $N_0\in\N$ such that, for every $g$ in $G$,
\[
(h\J_A)^{\star 2N_0}(g) \geq \bar c 
\]
where $h$ denotes the density of the absolutely continuous component of $\beta$, 
and $\star$ is the convolution of functions. 
Here clearly $N_0$ can be replaced by any power of two that is larger or equal,
at the possible expense of diminishing $\bar c$ to another (still positive) constant $c$.
This obviously implies our assertion \eqref{eq:bhatta}.
\begin{lemma}\label{Lemmapower}  
  There are $\tilde\kappa^\pm<1$ and $n\geq 1$ such that
  \begin{align*}
    \|[(L^\pm)^* L^\pm]^{2^n}f\|_{L^2} \le \tilde\kappa^\pm \|f\|_{L^2}\quad
    \text{for every $f\in L^2(\sod)$ with}\quad \int_{\sod}f(\rota)\haar(\dn\rota)=0.
  \end{align*}
\end{lemma}
\begin{proof} 
  We follow the lines of the proof of Theorem 2 in \cite{Bhattacharya}. 
 Assumption (H3) implies that the probability measures $\tpb^\pm$s have  nonzero absolutely continuous component with respect to the Haar measure.
  Hence we can apply Lemma \ref{Bhatta2}. 
  If $(\tpb^\pm)^{\star 2^n}  =\haar$ then $\|[(L^\pm)^* L^\pm)]^{2^n} f \|_{L^2}=0$, and there is nothing to be proved.
  If instead $(\tpb^\pm)^{\star 2^n} \not =\haar$, then $c<1$ in \eqref{eq:bhatta},
  and hence one can write
  \[
  (\tpb^\pm)^{\star 2^n}=[(1-c)\Gamma+c\haar],
  \]
  where $\Gamma=(1-c)^{-1}((\tpb^\pm)^{\star 2^n} -c\haar)$ is a probability measure on $\sod$. 
  Since $f$ is such that $\int f(\CO)\haar(d\CO)=0$, then, using also Jensen inequality,
 \begin{align*}
    \|[(L^\pm)^*L^\pm)]^{2^n} f \|_{L^2}^2 
    & =\int \Big ( \int f(\rota B) (\tpb^\pm)^{\star 2^n}(\dn B)  \Big)^2  \haar(\dn \rota) \\
    & =(1-c)^{2} \int \Big ( \int f(\rota B)\Gamma (\dn B)  \Big)^2 \haar (\dn\rota) \\
    & \leq (1-c)^{2} \int \int f(\rota B)^2 \haar (\dn\rota) \Gamma(\dn B) 
    =  (1-c)^{2} \|f \|_{L^2}^2.
  \end{align*}
  This shows the desired inequality, with $\tilde\kappa^\pm=  (1-c)  < 1$.
\end{proof}
\begin{proof}[Proof of Propostion \ref{propD}]
  Observe that
  \begin{align*}
    \|(L^\pm)^* L^\pm f\|_{L^2}^2= \langle [(L^\pm)^* L^\pm]^2 f,f\rangle_{L^2} \le \|[(L^\pm)^* L^\pm ]^2f\|_{L^2}\|f\|_{L^2}
  \end{align*}
  by the symmetry of $(L^\pm)^* L^\pm$.
  Similarly, for every $m\ge0$, we have
  \begin{align*}
    \big\|[(L^\pm)^* L^\pm]^{2^m} f\big\|_{L^2}^2 \le \big\|[(L^\pm)^* L^\pm]^{2^{m+1}} f\big\|_{L^2}^2\|f\|_{L^2},
  \end{align*}
  and iteration of these estimates leads to
  \begin{align*}
     \big\|(L^\pm)^* L^\pm f \big\|_{L^2}^{2^n} \le \big\|[(L^\pm)^* L^\pm]^{2^n}f \big\|_{L^2}\|f\|_{L^2}^{2^n-1}
  \end{align*}
  for arbitrary $n\ge0$.
  We combine this estimate with
 \begin{align*}
    \|L^\pm f\|_{L^2}^2 =  \langle (L^\pm)^* L^\pm f,f\rangle_{L^2} \le \|(L^\pm)^* L^\pm f\|_{L^2}\|f\|_{L^2}
  \end{align*}
  to obtain
  \begin{align*}
    \|L^\pm f\|_{L^2}^{2^{n+1}} \leq \|(L^\pm)^* L^\pm f\|_{L^2}^{2^n}\|f\|^{2^n}_{L^2}\leq \|[(L^\pm)^* L^\pm]^{2^n}f\|_{L^2}\|f\|_{L^2}^{2^n-1}\|f\|^{2^n}_{L^2}.
  \end{align*}
  Thus, by Lemma \ref{Lemmapower}, we arrive at
  \begin{align*}
    \|L^\pm f\|_{L^2}^{2^{n+1}} \le \tilde\kappa^\pm \|f\|_{L^2}^{2^{n+1}}.
  \end{align*}
  Taking the $2^{n+1}$th root, the hypothesis follows
  with $\kappa^\pm:=(\tilde\kappa^\pm)^{1/2^{n+1}}<1$.
\end{proof}
\begin{remark}
Note that in order to prove {\rm Proposition \ref{propD}} one only need the assumption
\begin{itemize}
\item[(H3')]
  \textit{The probability measures $\tilde \pb^\pm $ are non-singular with respect to the Haar measure, i.e.
    they have  a non-trivial absolutely continuous component with respect to $\haar$.} 
\end{itemize}
As a consequence, Theorem \ref{thm.main} holds under the weaker assumption {\rm (H3')} instead of {\rm (H3)}.
\end{remark}

\subsection{Convergence of Fourier transforms}
Out of the Fourier distance $d_\gamma$, we define yet another distance on $\CP(C^0[\sod])$ by
\[
D_{\gamma,a}(\nu',\nu'') := d_\gamma(\nu',\nu'')+ a \| \E[\Psi'-\Psi'']  \|_{L^2}
\]
where $\Psi'$, $\Psi''$ have law $\nu'$, $\nu''$, respectively.
Here $a$ is a positive constant to be determined later. 
Clearly, this distance satisfies the convexity inequality
\begin{equation}\label{convexity2}
 D_{\gamma,a} \Big (\frac1n\sum_{i=1}^n \mu_i',\frac1n\sum_{i=1}^n \mu_i'' \Big ) \leq \frac1n\sum_{i=1}^n D_{\gamma,a}(\mu_i', \mu_i'').
\end{equation}
\begin{proposition}\label{Dnto0} %
  $D_{\gamma,a}(\nu_n,\nu_\infty) \to 0$ as $n\to \infty$ for an appropriate choice of $\gamma>1$, $a>0$.
\end{proposition}
\begin{proof} 
  We are going to show that
  \begin{align}
    \label{eq:D2zero}
    D_{\gamma,a}(\nu_n,\nu_\infty) \leq \frac{\lambda}{n} \sum_{k=1}^n D_{\gamma,a}(\nu_{k-1},\nu_\infty)
  \end{align}
  with some $\lambda\in(0,1)$ for all $n\ge1$.
  This implies that $D_{\gamma,a}(\nu_n,\nu_\infty)\to0$ provided that
\begin{equation}\label{finite}
D_{\gamma,a}(\nu_0,\nu_\infty)<+\infty.
\end{equation}
Applying  triangular inequality and (i) in Lemma \ref{Lemmaweights} one gets
\[
 \| \E[\Psi_0-\Psi_\infty]  \|_{L^2} \leq \|\Psi_0\|_\infty + m_0.
\]
Moreover, since $|e^{ix}-1-ix| \leq C_\gamma |x|^\gamma$ for every $\gamma \in [1,2]$ (see e.g. Lemma 1, section 8.4 in \cite{ChowTeicher}),
\[
 d_{\gamma}(\nu_0,\nu_\infty)  \leq
C_\gamma \Big \{ \sup_{g : \|g\|_{L^2} \not=0}
\frac{1}{\|g\|_{L^2}^\gamma} \E[|\spr{g}{\Psi_0} |^\gamma+|\spr{g}{\Psi_\infty}|^\gamma] \Big \}.
\]
By the Cauchy-Schwarz inequality
\[
 d_{\gamma}(\nu_0,\nu_\infty) \leq \|\Psi_0\|_{L^2}^\gamma+ m_0^\gamma \E[(M_\infty^{(\a)})^\gamma].
\]
The last term is finite by (ii) of Lemma \ref{Lemmaweights} since in view of (H2) there exists some $\gamma \in (1,2]$ such that $\E[(r^-)^{\alpha\gamma}]+\E[(r^+)^{\alpha\gamma}]<1$. Now substitute \eqref{eq:10} and  \eqref{nuinftyfixT} into \eqref{convexity2} to obtain
  \[
  D_{\gamma,a}(\nu_n,\nu_\infty) \leq \frac{1}{n} \sum_{k=1}^n D_{\gamma,a}( T(\nu_{k-1},\nu_{n-k}), T(\nu_\infty,\nu_\infty)).
  \]
  Using the definitions of $T$ and of $D_{\gamma,a}$, 
  the terms on the right-hand side can be estimated as follows:
  \begin{align}
    \label{eq:Dhelp}
    \begin{split}
      D_{\gamma,a}( T(\nu_{k-1} ,\nu_{n-k}), T(\nu_\infty,\nu_\infty))  
      &\leq  d_\gamma( T(\nu_{k-1} ,\nu_{n-k}), T(\nu_\infty,\nu_\infty))  \\
      & + a\Big (\|\E[(r^-)^\alpha (R^-)^\star \Delta_{k-1} \|_{L^2}  + \|\E[(r^+)^\alpha (R^+)^\star \Delta_{n-k} \|_{L^2} \Big)
    \end{split}
  \end{align}
  where $\Delta_k:=\E[\Psi_k-\Psi_\infty]\in L^2(\sod;\haar)$ satisfies
  \begin{align*}
    \int_\sod \Delta_k(\CO) \haar(\dn \CO)=0 \quad \text{for every $k$},
  \end{align*}
  thanks to \eqref{eq:m0const}. 
  Hence Proposition \ref{propD} is applicable to estimate the last term on the right-hand side in \eqref{eq:Dhelp}.
  In combination with an estimate of the first term by means of Lemma \ref{lemmacontraction} -- which applies because of \eqref{boundL2normPsi} -- we arrive at 
 \[
  D_{\gamma,a}(\nu_n,\nu_\infty)   \leq \frac{1}{n} \sum_{k=1}^n\Big [ \lambda_\gamma d_\gamma(\nu_{k-1},\nu_\infty)
  + [a(\kappa^-+\kappa^+)+2C'] \|\E[\Psi_{k-1}-\Psi_\infty ]\|_{L^2} \Big ]
  \]
  with $\lambda_\gamma:=\E[(r^-)^{\alpha\gamma}]+\E[(r^+)^{\alpha\gamma}]<1$ and $C':=\max\{2,\|\Psi_0\|_\infty\}\E[(r^-)^{\alpha\gamma}+(r^+)^{\alpha\gamma}]$.
  Further, recalling that $\kappa^-+\kappa^+<\E[(r^-)^{\alpha}]+\E[(r^+)^{\alpha}]=1$,
  we can choose $a>0$ such that $a(\kappa^-+\kappa^+)+2C' < a$.
  Thus we have shown \eqref{eq:D2zero}, with
  \[
  \lambda:=\max\{\lambda_\gamma,\kappa^-+\kappa^+ +2C'/a\}<1.
  \qedhere
  \]
\end{proof}

\subsection{Proof of Proposition \ref{convPsin}}
By Proposition \ref{Dnto0} one gets
\begin{equation}\label{charfunctL2}
  \hat \nu_n(g) \to \hat \nu_\infty(g)
\end{equation}
for every $g$ in $L^2(\sod;\haar)$.
According to Lemma \ref{tightnessPsin}, $(\Psi_n)_n$ is a tight sequence in $\czero$. 
Assume that a subsequence $\Psi_{n'}$ converges weakly in $\czero$ to a limit $Y$.
Since $ f \mapsto \exp\{ i \spr{g}{f}  \}$ is a continuous function on $\czero$ for any $g$ in $L^2$, 
one gets that $\hat \nu_{n'}(g) \to \E [e^{ i \spr{g}{Y}}]$, and hence
\[
\E[e^{ i \spr{g}{\Psi_\infty}}]=\E [e^{ i \spr{g}{Y}}]
\]
for every $g$ in $L^2$.
Using the previous identity it is easy to see that the finite dimensional law of $Y$ and $\Psi_\infty$ are the same and hence they have the same distribution as processes (see, e.g., Proposition 3.2 \cite{Kallenberg}).
The last part of the proof follows by the continuous mapping theorem, since point evaluation is a continuous functional on $C^0[\sod]$.

\section{Proof of the main theorem}\label{sec.proof}

\subsection{Preliminary weak convergence results}

Recall that we deal with initial conditions $\mu_0$ belonging to the NDA of a (full) $\alpha$-stable law with
L\'evy measure $\phi$. Let
$X_0$ be a random variable with probability distribution $\mu_0$.
Moreover, for every $x, u  \in \setR^d$, set $F_0(x,u)=\P\{u \cdot X_0 \leq x\}$, $ F_0(x^-,u)=\lim_{y \uparrow x}F_0(y,u)$ and
\[
B_x=\{y \in \setR^d: x \cdot y > 1\}.
\]

Let $\CB_n$ denote the $\s\!-\!$field generate by the $\ww_{j,n}$'s and $O_{j,n}$, i.e.
\[
 \CB_{n}=\sigma(O_{j,n},\ww_{j,n}: j=1,\dots,n+1).
\]
Moreover, given any
$\CO \in \sod$,
write
\[
\be_{j,n}:=\CO O_{j,n} \eed  \qquad \qquad j=1,\dots,n+1
\]
and, for every
$y>0$, define
\[
\begin{split}
Q_{1,n}(y)&:=
\frac{1}{y^{\alpha}} \sum_{j=1}^{n+1} P\Big \{ \ww_{j,n} \be_{j,n} \cdot X_0 \geq 1/y \Big |\CB_n \Big \}=\frac{1}{y^\a}\sum_{j=1}^{n+1}
\Big [1-F_0\Big (\Big(\frac{1}{y\ww_{j,n}}\Big)^-, \be_{j,n}\Big)\Big ]\\
Q_{2,n}(y)&:=
\frac{1}{y^{\alpha}}  \sum_{j=1}^{n+1} P\Big \{  \ww_{j,n} \be_{j,n} \cdot X_0 \leq -1/y \Big |\CB_n \Big \}=\frac{1}{y^{\alpha}}\sum_{j=1}^{n+1}F_0\Big(-\frac{1}{y\ww_{j,n}}, \be_{j,n}\Big). \\
\end{split}
\]
Observe that
by Lemma \ref{lemma2} in Appendix  it follows that
\[
 \lim_{y \downarrow 0} Q_{1,n}(y)=\sum_{j=1}^{n+1}  \ww_{j,n}^\alpha \phi(B_{\be_{j,n}}) \quad \text{and} \quad
\lim_{y \downarrow 0} Q_{2,n}(y)=\sum_{j=1}^{n+1}  \ww_{j,n}^\alpha \phi(B_{-\be_{j,n}}).
\]
Hence setting
\[
 (Q_{1,n}(0),Q_{2,n}(0)):=(\sum_{j=1}^{n+1}  \ww_{j,n}^\alpha \phi(B_{\be_{j,n}}),\sum_{j=1}^{n+1}  \ww_{j,n}^\alpha \phi(B_{-\be_{j,n}})),
\]
the random function $y \mapsto (Q_{1,n}(y),Q_{2,n}(y))$ is a {\it c\`adl\`ag} (i.e. right continuos with left-hand limits) function from $[0,+\infty)$ to $\setR^2$.
Since, clearly, all the finite dimensional components are measurable,  $(Q_{1,n},Q_{2,n})$ can be seen as process
taking values in the space $\D(\setR_+,\RE^2)$ of c\`adl\`ag functions with the Skorohod topology (see, e.g., \cite{J-S} and  Thm. 4.5 in \cite{Billi}).
Furthermore, given any $\gamma_0\in\RE^d$ and $\CO \in \sod$, 
define
\[
 Q_{3,n}:=\sum_{j=1}^{n+1}  \ww_{j,n}^\alpha \CO O_{j,n}\ee_d\cdot\gamma_0.
\]
Since $Q_{3,n}$ and $M_n^{(\alpha)}$ can be seen as constant random functions (w.r.t. $y$), then $(M_n^{(\alpha)},Q_{1,n},Q_{2,n},Q_{3,n})$ is a process taking values in
$\D(\setR_+,\RE^4)$.

\begin{proposition}\label{conv-law} Assume {\rm (H1)-(H3)}.
 The sequence of processes  $(M_n^{(\alpha)}, Q_{1,n},Q_{2,n}, Q_{3,n})_{n\geq 1}$ converges in law in $\D(\setR_+,\RE^4)$ to 
the constant process $(M_\infty^{(\alpha)},c M_\infty^{(\alpha)},c M_\infty^{(\alpha)},0)$
where
\begin{equation}\label{def:c}
 c:=\int_{\sphere}\int_{\{y: y\cdot s > 1\}} \phi(dy)  \uuu(ds).
\end{equation}
\end{proposition}

\begin{proof}
First of all note that the functions
\(  \CO \mapsto \int_{\{y:  y \cdot \CO \eed  >1\}} \phi(dy)
\) and \(  \CO \mapsto \int_{\{y:  y \cdot \CO \eed  <-1\}} \phi(dy)
\)
are uniformly continuous on \(\sod \). Indeed, we know from Lemma \ref{lemma1bis} in Appendix  that
\(
 x \mapsto \int_{\{y:  y \cdot x  >1\}} \phi(dy)
\)
is continuous in $\setR^d\setminus \{0\}$. Hence it is uniformly continuous on $\sphere$ and the continuity of  $\CO \mapsto \CO \eed $ and $\CO \mapsto -\CO \eed $ entails the claim.

Now write for $i=1,2$
\[
Q_{i,n}(y)=Q_{i,n}(0)+R_{i,n}(y),
\]
and observe that, for $0<y \leq \delta$, one has
\[
\begin{split}
|R_{1,n}(y)|& \leq \sum_{j=1}^{n+1} \ww_{j,n}^\alpha \sup_u | (1/y\ww_{j,n})^\alpha(1-F_0((1/y\ww_{j,n})^-,u))- \phi(B_u)| \\
& \leq M_n^{(\alpha)} \sup_{z \leq \delta \beta_{(n)}} \sup_u | z^{-\alpha}(1-F_0((1/z)^-,u))- \phi(B_u)|. \\
\end{split}
\]
Hence
\[
 \sup_{0 \leq y \leq \delta}  |R_{1,n}(y)|\leq M_n^{(\alpha)} \sup_{0<y<\delta \beta_{(n)}} \sup_u | y^{-\alpha}(1-F_0((1/y)^-,u))- \phi(B_u)|.
\]
Analogously
\[
 \sup_{0 \leq y \leq \delta}  |R_{2,n}(y)|\leq M_n^{(\alpha)} \sup_{0<y<\delta \beta_{(n)}} \sup_u | y^{-\alpha}F_0(-1/y,u))- \phi(B_{-u})|.
\]
Since $\beta_{(n)} \to 0$ in probability by (iii) of Lemma \ref{Lemmaweights}, using Lemma \ref{lemma2} one obtains that
 for every $\delta>0$ and every $\eps>0$
\begin{equation}\label{convtozeroRn}
 \lim_{n \to +\infty} \P\{ \sup_{0 \leq y \leq \delta}[|R_{1,n}(y)|+|R_{2,n}(y)|] >\eps\}=0.
\end{equation}
Now let $(t_0,t_1,t_2, t_3) \in \setR^4$ and consider
\[
 \Psi_0(\CO):=t_0+t_1\int_{\{y: y \cdot \CO \eed>1\}}\phi(dy)+t_2\int_{\{y: y \cdot \CO \eed<-1\}}\phi(dy) 
+t_3 \CO\ee_d\cdot \gamma_0
\]
which is a continuous function on $\sod$ by the considerations above. Then the corresponding $\Psi_n$, defined in \eqref{defPsin}, satisfies
\[
\Psi_n(\CO)=t_0 M_n^{(\alpha)}+t_1\sum_{j=1}^{n+1}  \ww_{j,n}^\alpha \phi(B_{\be_{j,n}})+ t_2 \sum_{j=1}^{n+1}  \ww_{j,n}^\alpha \phi(B_{-\be_{j,n}})+t_3 Q_{3,n}
\]
since $\be_{j,n}=\CO O_{j,n} \eed$.
At this stage observe that 
\[
\int_\sod \CO\ee_d\cdot\gamma_0 \haar(d\CO)=0
\]
and Proposition \ref{convPsin} yields that
$\Psi_n(\CO)$ converges in law to $(t_0+t_1c_1+t_2c_2)M_\infty^{(\alpha)}$ where
\[
c_1:=\int_{\sod} \int_{\{y: y \cdot \CO \eed>1\}}\phi(dy) \haar(d\CO), \quad c_2:=\int_{\sod}\int_{\{y: y \cdot \CO \eed<-1\}}\phi(dy)\haar(d\CO).
\]
Since $O \eed$  is uniformly distributed on $\sphere$ whenever $ O$ has Haar distribution on $\sod$ (see \eqref{unif-on-sphere}), then
\[
 c_1=\int_{\sphere}\int_{\{y: y\cdot s> 1\}} \phi(dy) \uuu(ds)=c=\int_{\sphere}\int_{\{y: y\cdot s< -1\}} \phi(dy) \uuu(ds)=c_2.
\]
This yields that the vector
\[
Z_n:=(M_n^{(\alpha)},\sum_{j=1}^{n+1}  \ww_{j,n}^\alpha \phi(B_{\be_{j,n}}), \sum_{j=1}^{n+1}  \ww_{j,n}^\alpha \phi(B_{-\be_{j,n}}),Q_{3,n})
\]
converges in law to $(M_\infty^{(\alpha)}, cM_\infty^{(\alpha)},cM_\infty^{(\alpha)},0)$.
Since
\[
(M_n^{(\alpha)},Q_{1,n}(y),Q_{2,n}(y),Q_{3,n}(\CO_\ee))= Z_n+(0,R_{1,n}(y),R_{2,n}(y),0)
\]
using \eqref{convtozeroRn} and 
Lemma 3.31 Chapter VI of \cite{J-S}
one obtains the thesis.
\end{proof}

\subsection{Proof of Theorem \ref{thm.main}.}
The proof is split into three steps.
In the first step we introduce a Skorohod-type representation which is inspired to the one
used in \cite{FoLaRe} as an essential ingredient to prove central limit theorem for array of partially exchangeable random variables.
This technique has been already employed in a fruitful way in the context of the asymptotic study of kinetic equations, see e.g. 
\cite{BaLaRe,PeRe,DoRe2,GabettaRegazziniCLT}.
In the second step we prove that the classical conditions for the convergence to a (one-dimensional) stable law hold almost surely in 
the Skorohod representation. 
In the third step we conclude the proof. 

\vspace{0.3cm}

{\bf Step 1: Skorohod representation}.
For every $n\geq 1$ and for $j>n+1$, let us define  $\ww_{j,n}=0$ and $\be_{j,n}=\ee_d$, 
while for $j\leq n+1$ they are defined as in the previous sections.

Let $\CB_n$ denote the $\s\!-\!$field generated by the $\ww_{j,n}$'s and $\be_{j,n}$'s, i.e. $\CB_n=\s(\ww_{j,n},\; \be_{j,n}\; j\geq 1)$. Let $\lambda_{j,n}$ denote the conditional law of $\ww_{j,n}\be_{j,n}\cdot X_j$ given $\CB_n$ and $\lambda_n$ the conditional law of $\displaystyle \sum_{j=1}^{n +1}\ww_{j,n}\be_{j,n}\cdot X_j $, given $\CB_n$. Hence, $\displaystyle \lambda_{j,n}(-\infty, x]=F_0({x}/{\ww_{j,n}},\be_{j,n})$ and $\lambda_n=\lambda_{1,n}*\dots *\lambda_{n+1,n}$.
Let $Q_{3,n}=\sum_{j=1}^{n +1}\ww_{j,n}\be_{j,n}\cdot\gamma_0$ with $\gamma_0$ as in Theorem \ref{thm.main}  if $\a=1$ and with $\gamma_0=\mathbf 0$ otherwise.
Let us consider
\[
W_n=\Big(\lambda_n,(\lambda_{j,n})_{j\geq 1},\beta_{(n)},(\ww_{j,n})_{j\geq 1}, (\be_{j,n})_{j\geq 1},M_n^{(\a)}, Q_{1,n}(\cdot), Q_{2,n}(\cdot) ,Q_{3,n} \Big)
\]
as a random element from $(\Omega, \CF, P)$ in $(S,\CB(S))$, where
\(
S:=\CP (\bar{\setR} )^{\infty}\times \bar{\setR} ^{\infty}\times (\sphere )^{\infty}\times\D(\setR_+,\RE^4).
\)
Here $\bar{\setR}$ denotes the extended real line, $\CP (\bar{\setR})$ the set of all probability measures on borel $\sigma$-field $\CB(\bar{\setR})$ with the topology of the complete convergence and  $\CB(S)$ denotes the borel $\s-$field on $S$.

The sequence $(W_n)_{n\geq 1}$ is tight since $\CP (\bar{\setR} )^{\infty}$ and $(\sphere )^{\infty}$ are compact, $\beta_{(n)}\rightarrow 0$ in probability by Lemma \ref{Lemmaweights} and the sequence $(M_n^{(\a)}, Q_{1,n}(\cdot), Q_{2,n}(\cdot), Q_{3,n})_{n\geq 1}$
of random elements in $\D(\setR_+,\RE^4)$ converges in law to $(M_\infty^{(\alpha)},cM_\infty^{(\alpha)},cM_\infty^{(\alpha)},0)$ in view of Proposition \ref{conv-law}. Hence, every subsequence of $(n)$ includes a subsequence $(n')$ such that
\[
W_{n'}\stackrel{\CL}{\rightarrow} W'_{\infty}.
\]
Since $S$ is Polish, from the Skorohod representation theorem (see, e.g., Theorem 4.30 \cite{Kallenberg}) one can determine a probability space
$(\hat{\Omega}, \hat{\CF}, \hat{\P})$ and random elements on it taking value in $S$,
\[
\hat W_{\infty}=\Big(\hat \lambda,(\lambda_{j})_{j\geq 1},\hat\beta, (\hat\ww_{j})_{j\geq 0}, (\hat \be_{j})_{j\geq 1},\hat M, \hat Q_{1}(\cdot), \hat Q_{2}(\cdot),\hat Q_{3} \Big)
\]
\[
\hat W_{n'}=\Big(\hat \lambda_{n'},(\hat \lambda_{j,n'})_{j\geq 1},\hat \beta_{({n'})},(\hat\ww_{j,n'})_{j\geq 1}, (\hat\be_{j,n'})_{j\geq 1},\hat M_{n'}, \hat Q_{1,{n'}}(\cdot), \hat Q_{2,{n'}}(\cdot) , \hat Q_{3,{n'}}\Big)
\]
which have the same probability distribution of $W_\infty'$ and $W_\np$, respectively and
\[
\lim_{\np \rightarrow +\infty}\hat W_\np (\hat \omega)= \hat W_{\infty}(\hat \omega)
\]
for every $\hat \omega \in \hat \Omega$ in the metric of $S$. In view of the definition of $W_n$ and since $W_n$ and $\hat W_{n'}$ have the same probability distribution, the following statements hold, for each $\np$, $\hat \P-$a.s.
\begin{equation}\label{skr}
\begin{split}
&\hat \lambda_{n'}=\hat \lambda_{1,{n'}}*\dots *\hat \lambda_{\np+1,\np} , \qquad
\hat \lambda_{j,\np}(-\infty , x]=F_0(x/\hat\ww_{j,\np}, \hat\be_{j,\np}),\\
&\displaystyle \hat\beta_{(n')}=\max_{j=1,\dots,n'+1} \hat\beta_{j,\np}, \quad \hat M_\np=\sum_{j=1}^{\np+1}\hat \beta_{j,\np}^\a, \\
&\displaystyle \hat Q_{1,{n'}}(y)=\frac{1}{y^\a}\sum_{j=1}^{\np+1}[1-F_0((1/y\hat\ww_{j,\np})^-, \hat\be_{j,\np})]\text{ for every } y>0, \\
&\displaystyle \hat Q_{2,{n'}}(y)=\frac{1}{y^{\alpha}}\sum_{j=1}^{\np+1}F_0(-1/y\hat\ww_{j,\np}, \hat\ww_{j,\np})\text{ for every } y>0, \\
&\hat Q_{3,{n'}}=\sum_{j=1}^{\np+1}\hat \beta_{j,\np}^\a \hat\be_{j,\np}\cdot\gamma_0.
\end{split}
\end{equation}
Furthermore, since
\[
(\beta_{(n)},M_n^{(\a)}, Q_{1,n}(\cdot), Q_{2,n}(\cdot),  Q_{3,{n}})\stackrel{\CL}{\rightarrow} (0,M_\infty^{(\alpha)},cM_\infty^{(\alpha)},cM_\infty^{(\alpha)},0)
\]
then
\[
(\hat \beta,\hat M, \hat Q_{1}(\cdot), \hat Q_{2}(\cdot) ,\hat Q_{3})= (0, \hat M,c\hat M_,c\hat M,0) \]
$\hat \P-$a.s.. and the law of $\hat M$ is equal to the law of $M_\infty^{(\alpha)}$ and hence does not depend on the sequence $(\np)$.

{\bf Step 2: sufficient conditions for the convergence to a stable law}.
The next step is to prove  that the following conditions hold $\hat \P-$a.s.:
\begin{enumerate}
\item[$i)$] $\displaystyle \sup_{1\leq j\leq \np}\hat \lambda_{j,\np}([-\epsilon, \epsilon]^c)\to 0 $, for every $\epsilon >0$, as $\np\to +\infty$ (u.a.n. condition);
\item[$ii)$] $\displaystyle \lim_{n'\to +\infty} x^\a \sum_{j=1}^{\np+1}\hat \lambda_{j,\np}((-\infty, -x])= c{\hat M}$ and $\displaystyle \lim_{n'\to +\infty}x^\a \sum_{j=1}^{\np+1}\hat \lambda_{j,\np}(( x,+\infty])= c{\hat M}$  for every $x >0$, with $c$ as in \eqref{def:c};
\item[$iii)$]$\displaystyle \lim_{\epsilon\downarrow 0}\limsup_{\np\to +\infty}\sum_{j=1}^{\np+1}\int_{(-\epsilon,\epsilon)}x^2\hat \lambda_{j,\np}(dx)=0$;
\item[$iv)$]$\displaystyle \lim_{\np\to +\infty}E_\np=0$ where \[E_\np:=\Big\{\! -\!\sum_{j=1}^{\np+1}\hat \lambda_{j,\np}((-\infty, -1])+\! \sum_{j=1}^{\np+1}(1-\hat\lambda_{j,\np}(-\infty, 1 ])+\sum_{j=1}^{\np+1}\int_{(-1,1]}\!\!x\hat \lambda_{j,\np}(dx)\Big\}.\]
\end{enumerate}
In view of the well-known criteria for the convergence to a (one-dimensional) stable law  -
see, e.g., Theorem 30 in Section 16.9 and in Proposition 11 in Section 17.2 of \cite{fristed} -
the previous conditions yield that $\hat \P-$a.s.
\begin{equation}\label{SkoConv}
\int e^{i\rho x}\hat\lambda_\np (dx)\to \int e^{i\rho x}\hat\lambda (dx)=e^{-c\hat M |\rho|^\a}
\end{equation}
and this will lead easily to the conclusion. 
\vspace{0.4cm}

Let us first prove $i)$.
Recall that from Lemma \ref{lemma2} we know that
\begin{equation}\label{equivPhi1}
\lim_{x \to +\infty} \{\sup_{u \in \sphere} |x^\alpha (1-F_0(x,u))-\phi(B_u)|
+\sup_{u \in \sphere} |x^\alpha F_0(-x,u)-\phi(B_{-u})|\}=0
\end{equation}
and hence, in particular,
\begin{equation}\label{equivPhi2}
\sup_{x >0} \{\sup_{u \in \sphere} |x^\alpha (1-F_0(x,u))-\phi(B_u)|
+\sup_{u \in \sphere} |x^\alpha F_0(-x,u)-\phi(B_{-u})|\}<K.
\end{equation}
Since for every $u \in \sphere$, one has $\phi(B_u)\leq \phi\{ y: |y|\geq 1 \}<+\infty$, then
\eqref{equivPhi2} yields
\begin{equation}\label{equivPhi3}
\sup_{x >0} \{\sup_{u \in \sphere} |x^\alpha (1-F_0(x,u))|+\sup_{u \in \sphere} |x^\alpha F_0(-x,u)|\}<K'.
\end{equation}
In view of \eqref{skr} we have
\begin{align*}
\hat \lambda_{j,\np}([-\epsilon, \epsilon]^c)&\leq 1-F_0(\epsilon /\hat\ww_{j,\np}, \hat\be_{j,\np})+F_0(-\epsilon /\hat\ww_{j,\np}, \hat\be_{j,\np})\\
& \leq \frac{\hat\ww_{j,\np}^\a}{\epsilon^\a} \sup_{u \in \sphere} \Big\{ \Big[1-F_0(\epsilon /\hat\ww_{j,\np}, u)+F_0(-\epsilon /\hat\ww_{j,\np},u)\Big] \frac{\epsilon^\a}{\hat\ww_{j,\np}^\a}\Big\}\\
&\leq K' \frac{\hat\ww_{j,\np}^\a}{\epsilon^\a}\leq K' \frac{\hat\ww_{(n')}^\a}{\epsilon^\a}
\end{align*}
and the last term converges to zero for $\np\to +\infty$.

As for $ii)$, if $x>0$
\[
\displaystyle x^\a \sum_{j=1}^{\np+1}\hat \lambda_{j,\np}((x,+\infty))
=\hat Q_{1,\np}\Big(\Big(\frac{1}{x}\Big)^{-}\Big)
\quad \text{and} \quad
x^\a \sum_{j=1}^{\np+1}\hat \lambda_{j,\np}((-\infty, -x])
=\hat Q_{2,\np}\Big(\frac{1}{x}\Big).
\]
Since $(\hat Q_{1,\np}(\cdot),\hat Q_{2,\np}(\cdot))$ converges for every $\hat\omega \in \Omega$ in the topology of $\D(\setR_+,\RE^2)$ to the constant function $(c\hat M, c\hat M)$ then,
by using Proposition 2.4 Chapter VI of \cite{J-S}, one gets for every $y>0$
\[
\hat Q_{1,\np}(y^-)\to c\hat M \quad \text{and} \quad  \hat Q_{2,\np}(y)\to c\hat M.
\]
Hence $ii)$ is proved.

In order to prove $iii)$ note that integration by parts, gives
\[
\int_{(-\eps,\eps)}x^2dG(x) \leq 2\int_{-\eps}^0|x|G(x)dx+2\int_0^{\eps}x(1-G(x))dx.
\]
Hence the last inequality and \eqref{equivPhi3} yield
\begin{align*}
\sum_{j=1}^{\np+1}\int_{(-\epsilon,\epsilon)}x^2\hat \lambda_{j,\np}(dx)&=\sum_{j=1}^{\np+1}\hat \ww_{j,\np}^2\int_{(-\frac{\epsilon}{\ww_{j,\np}},\frac{\epsilon}{\ww_{j,\np}})}x^2F_0(dx, \hat\be_{j,\np})\\
& \leq
    2 \sum_{j=1}^{n'+1} \Big \{
    \int_{-\epsilon}^0 |x| F_0\Big (\frac{x}{\hat\ww_{j,{n'}}},\hat\be_{j,{n'}}\Big ) dx +
     \int_0^\epsilon x \big[1-F_0\Big (\frac{x}{\hat \ww_{j,{n'}}},\hat\be_{j,{n'}}\Big ) \big]\,dx
     \Big \} \\
    & \leq2\sum_{j=1}^{n'+1}\hat\ww_{j,{n'}}^\alpha K' \int_0^\epsilon x^{1-\alpha} dx
    =\hat M_{n'} \frac{2K' \epsilon^{2-\alpha}}{2-\alpha}
\end{align*}
which gives the result since $\hat M_{n'}$ converges $\hat \P$-a.s. to $\hat M$.

Concerning  $iv)$, assume first that $\alpha<1$.
Then,
integration by parts gives
\begin{align*}
E_{n'}=&-\sum_{j=1}^{\np+1}\int_{(-1,0]}\hat \lambda_{j,\np}((-\infty ,x])dx+\sum_{j=1}^{\np+1}\int_{(0,1]}(1-\hat \lambda_{j,\np}((-\infty ,x]))dx\\
\leq & -\sum_{j=1}^{\np+1}\int_{(-1,0]}F_0\Big (\frac{x}{\hat\ww_{j,{n'}}},\hat\be_{j,{n'}}\Big )dx+\sum_{j=1}^{\np+1}\int_{(0,1]}(1-F_0\Big (\frac{x}{\hat\ww_{j,{n'}}},\hat\be_{j,{n'}}\Big )dx.
\end{align*}
We know that, if $x<0$, $ -|x|^\alpha\sum_{j=1}^{\np+1}F_0\Big (\frac{x}{\hat\ww_{j,{n'}}},\hat\be_{j,{n'}}\Big )=\hat Q_{1,\np}(1/|x|)\to -c \hat M$
on $\hat \Omega$. Furthermore for every $x\in (0,1)$
\[
\Big|\sum_{j=1}^{\np+1}F_0\Big (\frac{x}{\hat\ww_{j,{n'}}},\hat\be_{j,{n'}}\Big )\Big| \leq \frac{1}{|x|^\alpha} \sup_{y <0} \sup_{u \in \sphere} F_0(y,u) |y|^\alpha \hat M_\np
\leq \frac{1}{|x|^\alpha} K' \sup_{n'} \hat M_\np.
\]
Analogously, for $x>0$, $|x|^\alpha\sum_{j=1}^{\np+1}\Big(1-F_0\Big (\frac{x}{\hat\ww_{j,{n'}}},\hat\be_{j,{n'}}\Big )\Big )=\hat Q_{2,\np}(1/|x|)\to c \hat M$
on $\hat \Omega$. Finally, for every $x\in (-1,0)$
\[
\Big|\sum_{j=1}^{\np+1}\Big(1-F_0\Big (\frac{x}{\hat\ww_{j,{n'}}},\hat\be_{j,{n'}}\Big )\Big)\Big| \leq \frac{1}{|x|^\alpha} \sup_{y <0} \sup_{u \in \sphere} (1-F_0(y,u)) |y|^\alpha \hat M_\np
\leq \frac{1}{|x|^\alpha} K' \sup_{n'} \hat M_\np.
\]
Hence, dominated convergence (for any $\hat \omega$) yields
\[
E_\np= -\sum_{j=1}^{\np+1}\int_{(-1,0]}\!\!\!\hat \lambda_{j,\np}((-\infty ,x])dx+\sum_{j=1}^{\np+1}\int_{(0,1]}\!\!\!(1-\hat \lambda_{j,\np}((-\infty ,x]))dx \to \!-\!\int_{(-1,0]}\!\! \frac{c \hat M}{|x|^\alpha}  dx
+\int_{(0,1]}\!\! \frac{c \hat M}{|x|^\alpha} dx=0.
\]
When $1<\a<2$, since $\int_{\setR} y F_0(dy,u)=0$ for every $u$ in $\sphere$, we can write
\[
E_{n'}= -\int_{(-\infty,-1]}(1+x)\sum_{j=1}^{\np+1} \hat \lambda_{j,\np}(dx)- \int_{(1,+\infty)}(x-1)\sum_{j=1}^{\np+1} \hat \lambda_{j,\np}(dx).
\]
Integration by parts gives
\[
\begin{split}
\int_{(-\infty,-1]} (1+x)\sum_{j=1}^{\np+1}\hat \lambda_{j,\np}(dx)
& =\lim_{T \to + \infty} \int_{(-T ,-1]} (1+x)\sum_{j=1}^{\np+1} \hat\lambda_{j,\np}(dx) \\
& = \lim_{T  \to + \infty} \Big[-\sum_{j=1}^{\np+1} \hat\lambda_{j,\np}((-\infty , 1-T ])(1- T)-\int_{(-T ,-1]}\sum_{j=1}^{\np+1} \hat\lambda_{j,\np}((-\infty, x])dx.\Big ]\\
\end{split}
\]
Now
\[
\limsup_{T \to + \infty }\sum_{j=1}^{\np+1} \hat\lambda_{j,\np}((-\infty, 1-T])(1-T)\leq \limsup_{T \to + \infty }
K M_{n}^{(\alpha)}(1- T)^{1-\alpha} =0,
\]
and hence
\[
\int_{(-\infty,-1]} (1+x)\sum_{j=1}^{\np+1} \hat\lambda_{j,\np}(dx)=-\int_{(-\infty ,-1]}\sum_{j=1}^{\np+1} \hat\lambda_{j,\np}((-\infty, x])dx.
\]
In an analogous way one shows that
\[
\int_{(1,+\infty)}(x-1) \sum_{j=1}^{\np+1}\hat \lambda_{j,\np}(dx)=\int_{(1,+\infty)} \sum_{j=1}^{\np+1} (1-\hat\lambda_{j,\np}((-\infty, x])dx,
\]
so that
\[
E_{n'}=\int_{(-\infty ,-1]}\sum_{j=1}^{\np+1} \hat\lambda_{j,\np}((-\infty, x]) dx  -\int_{(-\infty ,-1]}\sum_{j=1}^{\np+1} (1-\hat\lambda_{j,\np}((-\infty, x]))dx.
\]
Arguing as in the case $\a<1$ one proves that
\[
E_{n'} \to  +\int_{(-\infty,-1]} \frac{c \hat M}{|x|^\alpha}  dx
-\int_{(1,+\infty)} \frac{c\hat M}{|x|^\alpha} dx=0.
\] 
It remains to consider the case $\alpha=1$. Note that by point $ii)$ with $x=1$ 
\[
\lim_{\np\to +\infty}E_\np=
\lim_{\np\to +\infty}
 \sum_{j=1}^{\np+1}\int_{(-1,1]}\!\!x\hat \lambda_{j,\np}(dx)
\]
if the limit exists and 
\[
\begin{split}
 \sum_{j=1}^{\np+1}\int_{(-1,1]}\!\!x\hat \lambda_{j,\np}(dx)
 & =\sum_{j=1}^{\np+1} \hat \beta_{j,n'} \Big [ \int_{(-1/\hat \beta_{j,n'},1/\hat \beta_{j,n'}]} x dF_0(x,\hat\be_{j,{n'}})-\gamma_0\cdot 
\hat\be_{j,{n'}} \Big ] +\sum_{j=1}^{\np+1} \hat \beta_{j,n'}\gamma_0\cdot 
\hat\be_{j,{n'}} \\
& =:E_\np^*+\sum_{j=1}^{\np+1} \hat \beta_{j,n'}\gamma_0\cdot \hat\be_{j,{n'}}\\
&= E_\np^*+ \hat Q_{3,n}
\\
\end{split}
\]
and $\hat Q_{3,n}\to 0\; \hat \P$-a.s.. Furthermore
\[
 | E_\np^*| \leq \hat M_{n'}  \sup_{R \geq 1/\hat \beta_{(n')}}\sup_{u \in \sphere} | \int_{(-R,R]}   x dF_0(x,u)-\gamma_0 \cdot u|.
\]
Since $\beta_{(n')} \to 0$ and $\hat M_{n'} \to \hat M$ it follows from assumption \eqref{hy-alpha=1BIS}
that  $\lim_{\np\to +\infty}E_\np=0$ in the case $\a=1$ too.
At this stage the proof of $iv)$ is completed.

{\bf Step 3: conclusion of the proof}.
By \eqref{SkoConv} and dominated convergence theorem one has
\[
\begin{split}
\E[e^{i\rho\sum_{k=1}^{\np+1}
\ww_{k,\np} \be_{k,\np} \cdot X_k}]&=\hat \E [\int e^{i\rho x}\hat\lambda_\np (dx)]\\ &\to \hat \E[\int e^{i\rho x}\hat\lambda (dx)]=\hat \E[e^{-c\hat M |\rho|^\a}]=\E[e^{-c M^{(\a)}_\infty |\rho|^\a}]
\end{split}
\]
where $\hat \E$ denotes the expectation with respect to $\hat \P$ and the last equality is due to the fact that we proved that $M^{(\a)}_\infty$ and $\hat M$ have the same probability distribution. In particular we have stated that the limit does not depend on the subsequence $(\np)$ and hence the convergence is true for the entire sequence $(n)$. Hence, using also Proposition \ref{probrep},
one has that for every $\ee \in \sphere$ and any $\rho >0$
\[
\lim_{n \to \infty} \hat \mu_n(\rho \ee) = \E[e^{-c M^{(\a)}_\infty |\rho|^\a}].
\] 
At this stage, the convergence of $\mu(t)$ to $\mu_{\infty}^c$ follows from 
\eqref{Wild2zero}. In order to prove the last part of the theorem it is enough to check that since $\mu_{\infty}^c$ is a scale mixture of a spherically symmetric stable law, it belongs to NDA of the same stable law.

\appendix

\section{Multivariate stable laws and their domain of attractions}\label{app:stable}

A random vector $Z$ taking values in $\setR^d$ has  a centered $\alpha$-stable distribution, for $\alpha$ in $(0,2)$, if and only if its characteristic function is
\begin{equation}\label{defStabile}
\E[e^{i \xi \cdot Z}]=\exp\Big \{ -  \int_{\sphere} |\xi \cdot s|^{\alpha} \eta(\xi,s) \Lambda(ds)\Big \} \qquad (\xi \in \setR^d)
\end{equation}
where $\Lambda$ is a finite measure on $\sphere$ and
\[
\eta(\xi,s) :=
  \left \{
    \begin{array}{ll}
    1-i\text{sign}( \xi\cdot s) \tan(\frac{\pi \alpha}{2}) & \text{if $\alpha \not=1$} \\
     1+i \frac{2}{\pi} \text{sign}( \xi\cdot s)  \log|\xi \cdot s| & \text{if $\alpha =1$}.\\
    \end{array}
  \right .
\]
See, e.g., Theorem 7.3.16 in \cite{MeerSche}. 

A random vector $Z$ has a centered  $\alpha$-stable spherically symmetric distribution if
\begin{equation}\label{sphersymmchar}
\E[e^{i \xi \cdot Z}]=\exp\{ - c |\xi|^\alpha \} \qquad (\xi \in \setR^d)
\end{equation}
for some $c>0$. Clearly, in this case, $\Lambda(A) \propto |A|$. 

As in the one-dimensional  case, one says that:

\noindent {\it A random vector $X_0$ {\rm(}or equivalently its law $\mu_0${\rm)}
belongs to the normal domain of attraction {\rm(}$NDA$, for short{\rm)} of an
$\alpha$-stable law  if for any sequence $(X_i)_{i \geq 1}$
of i.i.d. random vectors with the same law of $X_0$, there is a sequence
of vectors $(b_n)_{n \geq 1}$ such that $n^{-1/\alpha} \sum_{i=1}^n X_i-b_n$
converges in law to an $\alpha$-stable random vector.}

Given any a finite measure $\Lambda$ on $\sphere$  the so-called {\it L\'evy measure $\phi=\phi_\Lambda$}
  on $\setR^d \setminus \{0\}$ is given in polar coordinates
by
\begin{equation}\label{defphi}
\phi(d\theta dr)= \Lambda(d\theta) \frac{\alpha k_\alpha}{r^{\alpha+1}}dr.
\end{equation}

A stable law is said to be  {\it full} if it is 
not  supported on any $d-1$ dimensional subspace of $\setR^d$.
In this case, it is possible to characterize the $NDA$ in terms of the tails of $\mu_0$ in the following way:

\noindent {\it
$X_0$ belongs to the $NDA$ of a  stable law with L\'evy measure $\phi=\phi_\Lambda$
 if and only if
for every $r>0$ and every Borel set $B \subset \sphere$ such that $\Lambda(\partial B)=0$
\begin{equation}\label{NDA}
 \lim_{t \to +\infty} t^\alpha \P\Big \{|X_0|> r t , \,\, \frac{X_0}{|X_0|} \in B \Big \}= \frac{k_\alpha}{ r^\alpha} \Lambda(B),
\end{equation}
 with
\[
k_\alpha=\frac{2\Gamma(\alpha)\sin(\alpha\pi/2)}{\pi}.
\]}
See Theorems 6.20 and 7.11 in \cite{AraujoGine}.

We collect some results on the NDA of an $\alpha$-stable law,
which are used in Section \ref{sec.proof}.
\begin{lemma}\label{lemma0}
If a stable law  is full, then the corresponding L\'evy measure
$\phi$ is full, that is  $\phi$ is not supported on any $d-1$ dimensional subspace of $\setR^d$. 
\end{lemma}

\begin{proof}
The thesis can be deduced  combining Proposition 3.1.20 and Theorem 7.3.3 in 
\cite{MeerSche}.
\end{proof}

Recall that, for every $x \in \setR^d$,
\[
B_x=\{y \in \setR^k: x \cdot y > 1\}.
\]

\begin{lemma}\label{lemma1bis}
 Let $\phi$ be a full L\'evy measure, then
\[
x \mapsto \phi(B_x)
\]
is a continuous function on $\setR^d\setminus \{0\}$.
\end{lemma}

\begin{proof}
The proof is essentially the same as the proof of Lemma  6.1.25 in  \cite{MeerSche} and it is left to the reader.
\end{proof}

\begin{lemma}\label{lemma2}
Let $S$ be a compact subset of $\setR^d \setminus \{0\}$.
If $X_0$ belongs to the normal domain of attraction of a full $\alpha$-stable law with
L\'evy measure $\phi$, then
\begin{equation}\label{unifcond}
\lim_{t \to +\infty} \sup_{u \in S} | t^\alpha \P\{ X_0 \cdot u >t\} - \phi(B_u)| =0
 \end{equation}
and
 \begin{equation}\label{unifcond2}
\lim_{t \to -\infty} \sup_{u \in S} | |t|^\alpha \P\{ X_0 \cdot u \leq t\} - \phi(B_{-u})| =0.
 \end{equation}
Moreover \eqref{unifcond} remains true if
one replace $>$ with $\geq$.
\end{lemma}

\begin{proof}
The proof of this result  can be obtained  with  minor modifications from the proof of a similar
result contained in Lemma 6.1.26  of \cite{MeerSche}.
The details are left to the reader. 
\end{proof}

\section*{Acknowledgments} 
The authors would like to thank E. Breuillard, E.Dolera and A. Ghigi for helpful comments.


\begin{thebibliography}{99}
\bibitem{alsmeyer}
{Alsmeyer, G.; Meiners, M.} (2013).
Fixed points of the smoothing transform: two-sided solutions. {\em Probab. Theory Relat. Fields}. {\bf 155} 165--199.

\bibitem{AaronsonDenker} 
{Aaronson, J.; Denker, M.} (1998).  Characteristic functions of random variables attracted to 1-stable laws. {\it Ann. Probab.} {\bf 26} 399--415.

\bibitem{AraujoGine}
{ Araujo, A.;  Gin\'e, E.} (1980).
{\em The central limit theorem for real and Banach valued random variables.}
 Wiley.


\bibitem{BaLa}
{Bassetti, F.; Ladelli, L.}  (2012).
Self similar solutions in one-dimensional kinetic models: a probabilistic view.
 {\em Ann. App. Prob.} \textbf{22} 1928--1961.



\bibitem{BaLaMa}
{Bassetti, F.; Ladelli, L.; Matthes, D.} (2010).
Central limit theorem for a class of one-dimensional kinetic equations. {\em  Probab. Theory Related Fields} \textbf{150}  77--109.

\bibitem{BaLaRe} 
{Bassetti,  F.; Ladelli, L.; Regazzini, E.} (2008).
  {Probabilistic study of the speed of approach to equilibrium for an inelastic Kac model.}
  \emph{J. Stat. Phys. } {\bf 133} 683--710.


\bibitem{BaMa}
{Bassetti, F.; Matthes, D.} (2014).
Multi-dimensional smoothing transformations: Existence, regularity and stability of fixed points. 
\emph{Stochastic Process. Appl.} {\bf 124}
154--198

\bibitem{Bhattacharya} {Bhattacharya, R. N.} (1972). Speed of convergence of the n-fold convolution of a probability measure on a compact group. {\em Z. Wahrsch. Verw. Gebiete} \textbf{25} 1--10.

\bibitem{Billi}
  {Billingsley, P.} (1968). 
  \textit {Convergence of  Probability Measures}. 
  Wiley.

%

%
\bibitem{BiCaTo2}
  {Bisi, M.; Carrillo, J. A.; Toscani, G.} (2006).
  {Decay rates in probability metrics towards homogeneous cooling states for the inelastic Maxwell model.}
  \emph{J. Stat. Phys.} \textbf{124}, no. 2-4, 625--653.

\bibitem{Bobylev88}
  {Bobylev, A. V.} (1988).
  {The theory of the nonlinear spatially uniform Boltzmann equation for Maxwell molecules.} 
  \emph{Sov. Sci. Rev. C Math. Phys.} \textbf{7}, 111--233.

\bibitem{BoCe} {Bobylev, A. V.; Cercignani, C.} (2002).
Self-Similar Solutions of the Boltzmann Equation and Their Applications
\emph{ J. Statist. Phys.} 
{\bf 106} 1039--1071.

%
\bibitem{BoCeter}
  Bobylev, A. V.; Cercignani, C. (2003). 
  Self-similar asymptotics for the Boltzmann equation with inelastic and elastic interactions. 
  \emph{ J. Statist. Phys.} {\bf 110} 333--375. 

\bibitem{Bobylev99}
  {Bobylev, A.V.; Carrillo, J. A.; Gamba, I. M.} (1999). 
  {On Some Properties of Kinetic and Hydrodynamic Equations for Inelastic Interactions}.
  \emph{J. Statist. Phys.} {\bf 98}  {743--773}.

\bibitem{BoCeGa2}
  {Bobylev, A.V.; Cercignani, C.; Gamba, I. M.}  (2008).
  Generalized kinetic Maxwell type models of granular gases.
  In: Mathematical models of granular matter Series: Lecture Notes in Mathematics, Vol. 1937, G. Capriz,
  P. Giovine, P. M. Mariano (eds.) Berlin-Heidelberg-New York: Springer, 23--58.

%
\bibitem{BoCeGa}
  {Bobylev, A.V.; Cercignani, C.; Gamba, I. M.} (2009).
  On the self-similar asymptotics for generalized nonlinear kinetic maxwell models.
  \emph{Comm. Math. Phys.} \textbf{291} 599--644.

%
\bibitem{BoCeTo}
  Bobylev, A. V.; Cercignani, C.; Toscani, G. (2003). 
  Proof of an asymptotic property of self-similar solutions of the Boltzmann equation for granular materials. 
  \emph{J. Statist. Phys.} \textbf{111} 403--417. 

%
\bibitem{BolleyCarrillo}
  {Bolley, F.; Carrillo, J. A.} (2007).
  Tanaka theorem for inelastic Maxwell models. 
  \emph{Comm. Math. Phys.} \textbf{276} 287--314.



\bibitem{PeRe} 
Bonomi A.; Perversi E.; Regazzini E. (2013).
Probabilistic View of Explosion in an Inelastic Kac Model.
\texttt{arXiv:1306.0691}


\bibitem{Cannone}
  {Cannone, M.; Karch, G.} (2010).
  Infinite energy solutions to the homogeneous Boltzmann equation.
  \emph{Comm. Pure Appl. Math.} {\bf 63} 747--778.


\bibitem{CarlCarrCarv}
  {Carlen, E. A.; Carrillo, J. A.; Carvalho, M. C.}  (2009).
  {Strong convergence towards homogeneous cooling states for dissipative Maxwell models.} 
  \emph{Ann. Inst. H. Poincar\'e Anal. Non Lineaire} \textbf{26}, no. 5, 1675--1700.

\bibitem{CaGaRe}
  {Carlen, E.; Gabetta, E.; Regazzini, E.} (2007)
  On the rate of explosion for infinite energy solutions of the spatially homogeneous Boltzmann equation. 
  \emph{J. Stat. Phys.} {\bf 129 } 699--723.

\bibitem{CarrTosc}
  {Carrillo, J. A.; Toscani, G.} (2007).
  Contractive probability metrics and asymptotic behavior of dissipative kinetic equations.
  \emph{Riv. Mat. Univ. Parma} \textbf{6} 75--198.


\bibitem{ChowTeicher}
  {Chow, Y. S.; Teicher, H.} (1997). 
  \textit{Probability theory. Independence, interchangeability, martingales.} Third edition.  Springer-Verlag, New York.

\bibitem{DolGabReg} 
  {Dolera, E.; Gabetta, E.; Regazzini, E.} (2009).
  {Reaching the best possible rate of convergence to equilibrium for the solutions of Kac's equation via central limit theorem.}
  \emph{Ann. Appl. Probab.} \textbf{19}  186--209.

\bibitem{dolera2}
  {Dolera E.; Regazzini E.} (2010).
  {The role of the central limit theorem in discovering sharp rates of convergence to equilibrium for the solution of the Kac equation.} 
  \emph{Ann. Appl. Probab.} {\bf 20}  430--461.

\bibitem{DoRe2} 
  {Dolera E.; Regazzini E.} (2012).
  Proof of a McKean conjecture on the rate of convergence of Boltzmann-equation solutions. 
  \texttt{arXiv:1206.5147}

\bibitem{DurrettLiggett1983}
  {Durrett, R.; Liggett, T.M.}  (1983). 
  Fixed points of the smoothing transformation. 
  \emph{Z. Wahrsch. Verw. Gebiete} \textbf{64} 275--301.

\bibitem{ErnsBrit}
  {Ernst, M. H.; Brito, R.} (2002).
  {Scaling solutions of inelastic Boltzmann equations with over-populated high energy tails.}
  \emph{J. Statist. Phys.} \textbf{109}, no. 3-4, 407--432. 


\bibitem{FoLaRe}
  {Fortini, S.; Ladelli, L.; Regazzini, E.} (1996).
  A central limit problem for partially exchangeable random variables. \emph{Theory Probab. Appl.} {\bf 41} 224--246


\bibitem{fristed}
  {Fristedt, B.; Gray, L.} (1997).
  \emph{A modern approach to probability theory}.
  Probability and its Applications. Birkh\"auser Boston Inc., Boston, MA.

\bibitem{GabettaRegazziniCLT}
  {Gabetta, E.; Regazzini, E.}   (2008).
  {Central limit theorem for the solution of the Kac equation.} 
  \emph{ Ann. Appl. Probab.}  \textbf{18}  2320--2336.

\bibitem{GabettaRegazziniWM}
  {Gabetta, E.;  Regazzini, E.} (2010).
  {Central limit theorem for the solution of the Kac equation: Speed of approach to equilibrium in weak metrics.} 
  \emph{Probab Theory  Related Fields} {\bf 146}  451--480.

\bibitem{Wennberg}
  {Gabetta, G.; Toscani, G.; Wennberg, B.} (1995).
  {Metrics for probability distributions and the trend to equilibrium for solutions of the Boltzmann equation.} 
  \emph{J. Statist. Phys.} \textbf{81}, no. 5-6, 901--934.

\bibitem{hurwitz}
  {Hurwitz, A.} (1897).
  Ueber die Erzeugung der Invarianten durch Integration.
  \emph{Nach. Gesell. Wissen, G\"ottingen Math-Phys Klasse.} 71--90.


\bibitem{J-S}
  {Jacod, J.; Shiryaev, A.N.} (1987). 
  \emph{Limit Theorems for Stochastic Processes}. Springer-Verlag.

\bibitem{Kallenberg}
  {Kallenberg, O.} (2002). 
  \emph{Foundations of Modern Probability}. 
  Springer-Verlag.

\bibitem{McKean1966}
  {McKean Jr, H.~P. }  (1966).
  {Speed of approach to equilibrium for Kac's caricature of a Maxwellian gas.}
  \emph{Arch. Rational Mech. Anal.} \textbf{21}  343--367.

\bibitem{McKean1967} McKean Jr, H. P.  (1967).
  {An exponential formula for solving Boltzmann's equation for a Maxwellian gas.} 
   \emph{ J. Combinatorial Theory} {\bf 2}  358--382.  

\bibitem{MeerSche}
{ Meerschaert, M.M.; Scheffler, H.P.} (2001).
 \emph{ Limit distributions for sums of independent random vectors.}
  Wiley.


  %
  %

  %
  
  %
\bibitem{Mischler3}
  {Mischler,  S.; Mouhot, C.} (2009). 
  Stability, convergence to self-similarity and elastic limit for the Boltzmann equation for inelastic hard spheres. 
  {\em Commun. Math. Phys.} {\bf 288} 431--502. 

  %



\bibitem{rudin}
  {Rudin, W.} (1991). 
  \emph{Functional Analysis}. McGraw-Hill International Editions.

  %
\bibitem{PulvTosc}
  {Pulvirenti, A.; Toscani, G.} (2004).
  {Asymptotic properties of the inelastic Kac model.}
  \emph{J. Statist. Phys.} \textbf{114}, no. 5-6, 1453--1480.
  
  %
\bibitem{Tanaka}
  {Tanaka, H.} (1978/79)
  {Probabilistic treatment of the Boltzmann equation of Maxwellian molecules.} 
  \emph{Z. Wahrsch. Verw. Gebiete} \textbf{46}, no. 1, 67--105.

  %
\bibitem{ToscVill}
  {Toscani, G.; Villani, C.} (1999). 
  {Probability metrics and uniqueness of the solution to the Boltzmann equation for a Maxwell gas.} 
  \emph{J. Statist. Phys.} \textbf{94}, no. 3-4, 619--637.

\bibitem{Villani}
  {Villani, C.} (2002).
  \emph{A review of mathematical topics in collisional kinetic theory.} 
  Handbook of mathematical fluid dynamics, Vol. I, 71--305, North-Holland, Amsterdam.
\end{thebibliography}
\end{document}